\documentclass[10pt]{article}

\usepackage[T1]{fontenc}
\usepackage[letterpaper,textwidth=6.8in,textheight=9in,centering]{geometry}
\usepackage[parfill]{parskip}
\usepackage{amsmath,amssymb,amsthm,mathtools}
\usepackage[tt=false,type1=true]{libertine}
\usepackage[varqu]{zi4}
\usepackage[libertine]{newtxmath}
\usepackage{microtype,booktabs,tabularx,enumitem,needspace}
\usepackage{graphicx}
\usepackage{algorithm,algpseudocode}
\usepackage{xcolor}
\definecolor{PaperLinkBlue}{HTML}{2A7FB8}
\definecolor{PaperBlue}{HTML}{1F4E79}
\usepackage{titlesec}
\titleformat{\section}{\normalfont\Large\bfseries}{\thesection}{1em}{}
\titleformat{\subsection}{\normalfont\large\bfseries\color{PaperBlue}}{\thesubsection}{1em}{}
\titleformat{\subsubsection}{\normalfont\normalsize\bfseries\color{PaperBlue}}{\thesubsubsection}{1em}{}
\titleformat{\paragraph}[runin]{\normalfont\normalsize\bfseries\color{PaperBlue}}{\theparagraph}{1em}{}
\titlespacing*{\paragraph}{0pt}{1.7mm}{0.6em}
\usepackage[colorlinks=true,linkcolor=PaperLinkBlue,citecolor=PaperLinkBlue,urlcolor=PaperLinkBlue]{hyperref}
\newtheorem{theorem}{Theorem}[section]
\newtheorem{lemma}[theorem]{Lemma}
\newtheorem{proposition}[theorem]{Proposition}
\newtheorem{corollary}[theorem]{Corollary}
\newtheorem{assumption}[theorem]{Computational model}
\theoremstyle{definition}\newtheorem{definition}[theorem]{Definition}
\theoremstyle{remark}\newtheorem{remark}[theorem]{Remark}
\DeclareMathOperator{\Tr}{Tr}

\DeclareMathOperator{\Fid}{Fid}
\DeclareMathOperator{\diag}{diag}

\newcommand{\R}{\mathbb R}
\newcommand{\C}{\mathbb C}
\newcommand{\E}{\mathbb E}

\newcommand{\norm}[1]{\lVert#1\rVert}

\numberwithin{equation}{section}
\allowdisplaybreaks[2]

\title{A Walk From Free Probability to Matrix\\
Discrepancy II: Weaver's Problem and the Kadison--Singer\\
Conjecture}
\author{Tarun Kathuria\\[0.6ex]
\includegraphics[width=1.4in]{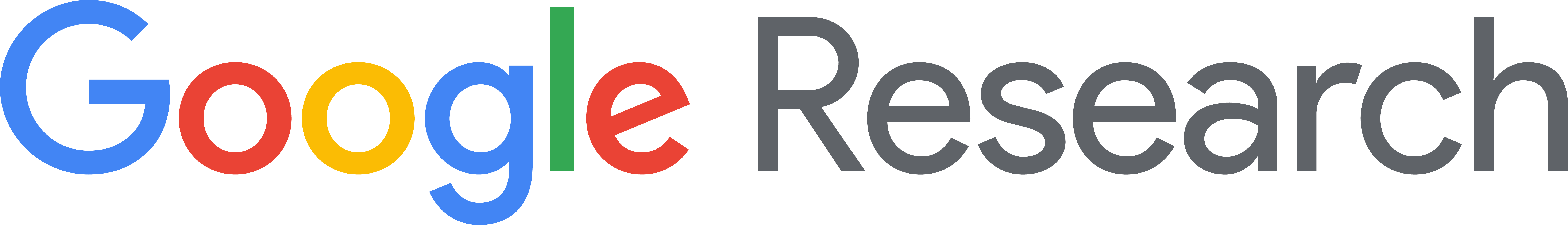}\\[0.6ex]
{\normalsize\href{mailto:tarunkathuria@google.com}{\texttt{tarunkathuria@google.com}}}}
\date{September 15, 2026}
\hypersetup{pdftitle={A Walk From Free Probability to Matrix Discrepancy II: Weaver’s Problem and the Kadison-Singer Conjecture},pdfauthor={Tarun Kathuria}}
\begin{document}
\maketitle
\begin{abstract}
Weaver's discrepancy problem asks whether vectors
$v_1,\ldots,v_n\in\C^m$ satisfying
$\sum_{i=1}^n v_iv_i^*=I_m$ and $\|v_i\|_2^2\le\varepsilon$
admit a signing $x\in\{-1,1\}^n$ such that
\[
 \left\|\sum_{i=1}^n x_iv_iv_i^*\right\|\le O(\sqrt\varepsilon).
\]
Marcus, Spielman, and Srivastava proved this existence result,
resolving the Kadison--Singer conjecture \cite{mss2015}. Finding such
signs efficiently for general inputs remained an algorithmic question.
In the real-arithmetic model, we give a deterministic algorithm running
in polynomial time with discrepancy at most $35\sqrt\varepsilon$.

The algorithm walks from the origin of the hypercube to a vertex,
fixing coordinates only when they are close to a face. As in our
companion paper on Matrix Spencer \cite{kathuria2026ms}, its potential
measures a soft spectral edge of the discrepancy matrix perturbed by an operator-valued
free semicircular element. The perturbation's covariance vanishes as
the coefficients reach their endpoints. Inspired by the free
interpolation approach of Bandeira, Boedihardjo, and van Handel
\cite{bbvh2023}, we combine Lehner's variational formula
\cite{lehner1999} with spectral Tsallis--$1/2$ regularization used in
spectral sparsification framed as matrix optimization
\cite{allenZhuLiaoOrecchia2015} and in discrepancy minimization
\cite{pesentivladu2026}. The resulting potential
has a finite-dimensional semidefinite formulation, allowing the
discrepancy and remaining covariance to be analyzed together.

We analyze the optimizer's stability through the linearized
Karush--Kuhn--Tucker (KKT) system of a regularized min--max problem,
whose stationarity equations are related to the matrix Dyson equation
\cite{erdos2019}. This gives the movement rule: either a
coordinate can move toward its nearer endpoint at small spectral cost,
or a low-curvature direction orthogonal to the current coefficient
vector allows further progress. Choosing the better sign of this
direction controls discrepancy while increasing the squared distance
from the origin. This yields a finite walk that reaches a signing
with the required bound.

Upcoming work \cite{kathuria2026higherRank} will address higher-rank
Kadison--Singer problems and spectrally thin trees.
Lean formalizations of our main discrepancy theorems have been
completed and will be released shortly.
\end{abstract}

\section{Introduction}
Let $v_1,\ldots,v_n\in\C^m$ satisfy
\[
 \sum_{i=1}^n v_iv_i^*=I_m,
 \qquad \|v_iv_i^*\|\le\varepsilon.
\]
We seek a partition of the vectors into two parts whose matrix sums
are close to $I_m/2$. Equivalently, writing $A_i=v_iv_i^*$, we seek
signs $\sigma_i$ for which $\|\sum_i\sigma_iA_i\|$ is small.
Weaver's discrepancy formulation connects this finite-dimensional
problem to the Kadison--Singer conjecture
\cite{kadison1959,weaver2004}. Marcus, Spielman, and Srivastava proved
it through interlacing polynomials \cite{mss2015}.
Here we construct the signs by a walk inside the coefficient cube.
Our algorithm runs in polynomial time in the real-arithmetic
model.\footnote{We allow scalar real arithmetic, comparisons,
nonnegative square roots, and exact symmetric eigendecomposition.
We also assume a deterministic SDP solver returning values to additive
accuracy $\nu$ in time polynomial in the scalar data size and
$\nu^{-1}$. Model~\ref{model:solver} gives the precise assumptions.}

\Needspace{15\baselineskip}
\begin{theorem}[A deterministic local signing walk]\label{thm:main}
Suppose $\sum_i v_iv_i^*=I_m$ and $\|v_iv_i^*\|\le\varepsilon$ for
all $i$, where $\varepsilon\ge0$. There is a deterministic algorithm in
Model~\ref{model:solver} with the following properties.
\begin{enumerate}[label=(\roman*),leftmargin=*]
\item It returns signs $\sigma\in\{-1,1\}^n$ satisfying
\[
 \left\|\sum_i\sigma_i v_iv_i^*\right\|\le35\sqrt\varepsilon.
\]
\item It runs in time polynomial in $n,m$ in this model.
\item If $m>0$, the coefficients start at zero and remain in $[-1,1]$.
 With $\rho=1/(100n^2)$, each movement before rounding changes each
 coordinate by at most $\rho/16$. A coordinate is assigned its final
 endpoint only within distance $\rho$ of that endpoint and then stays fixed.
\end{enumerate}
For $m=0$, the all-positive signing is returned immediately.
\end{theorem}

The signing partitions the labels into
$I_+=\{i:\sigma_i=1\}$ and $I_-=\{i:\sigma_i=-1\}$, with
\begin{equation}\label{eq:partition}
 \left\|\sum_{i\in I_\pm}A_i-\tfrac12I_m\right\|
 \le\frac{35}{2}\sqrt\varepsilon.
\end{equation}
Indeed, the two sums are $(I_m\pm\sum_i\sigma_iA_i)/2$.
The existence conclusion follows from the analytic argument alone
(Proposition~\ref{prop:exacttrial}); the computational model is used
for the runtime bound. Section~\ref{sec:sdp}
derives the queried programs, including their primal and dual values;
Appendix~\ref{sec:runtime} bounds the work of the complete deterministic
run.

\subsection{The idea of the walk}
At each state the walk tests whether a coefficient can move a short
distance toward its nearer endpoint at small spectral cost. It takes
the first affordable move. If every test fails, it finds a unit direction
$g$ orthogonal to the current coefficient vector $x$, with small
curvature of the potential. For a fixed small step size $h>0$, it compares the potential values at
$x+hg$ and $x-hg$ and takes the better reported candidate. Coordinates
near an endpoint are rounded, and the process repeats until all
coordinates are signs.

Why should failure of the coordinate tests produce a useful direction?
The potential measures the current discrepancy together with an
auxiliary source that disappears as coordinates reach the boundary.
Source removal can pay for an outward increase of the signed sum.
When none of these outward steps is affordable, the failures bound
the transport probes in the variational problem. These bounds show
that the source's negative curvature survives reoptimization of the
density. The proof averages over compatible coefficient and transport
responses. Imposing $x^Tg=0$ costs at most one further dimension in this
average; a fixed-point identity supplies enough negative mass to absorb
that loss.

Orthogonality supplies the progress that determines the choice of
sign: $\|x\pm hg\|_2^2=\|x\|_2^2+h^2$. Thus both candidates make the
same progress toward the cube's vertices. The average of their
potential changes has no linear term and is controlled by curvature;
choosing the lower report incurs only the prescribed value error.
Outward steps and rounding also increase $\|x\|_2^2$. A proof account combines the spectral potential, the remaining rounding
cost, and a small multiple of the deficit $n-\|x\|_2^2$. This account
never increases. It proves both a discrepancy bound
throughout the walk and completion after finitely many movements.

\subsection{Free interpolation and variational stability}
Bandeira, Boedihardjo, and van Handel use interpolation to compare
smooth moments and resolvent statistics of random matrices with their
operator-valued free counterparts; Lehner's formula gives a variational
description of the free norm
\cite[Section~1.4.2 and Lemma~2.4]{bbvh2023}.
Our construction places the covariance of that free model inside a
smooth variational potential. The discrepancy center and the source
then change together during the walk. The smoothing uses the
square-root density regularizer used to frame spectral sparsification
as matrix optimization
\cite{allenZhuLiaoOrecchia2015}, closely related to the regularized
maximum in Pesenti--Vladu's discrepancy method \cite{pesentivladu2026}.
Section~\ref{sec:potentialmotivation} explains why the potential retains
the full covariance map, rather than only a remaining-variance matrix.

The maximizing density and minimizing transport form a saddle point
of a regularized min--max problem. Their stationarity and normalization
equations are its Karush--Kuhn--Tucker (KKT) conditions. Differentiating
these conditions gives the linearized KKT system governing the
optimizer's response as the coefficients move.
The structural comparison is with matrix Dyson
equation stability \cite{aek2019}, particularly the analysis near a
regular spectral edge \cite[Section~4]{aeks2020} and the deterministic
theory of edges and cusps \cite{alt2020}. Our self-consistent equations
include a density regularizer and a trace constraint. We derive and
bound their response directly. In particular, the balanced fixed-point
equation has directions that cannot be controlled by simply inverting
its difference from the identity; the joint-kernel projection in
Section~\ref{sec:sm:curvature} is the substitute used here.

\subsection{Related discrepancy results}
The constant in Theorem~\ref{thm:main} is deliberately coarse. Besides
the sharper partition theorem of Marcus--Spielman--Srivastava,
Kyng--Luh--Song obtain a rank-one discrepancy bound of four times the
matrix standard deviation \cite{kyngLuhSong2020}.
For algorithms based on interlacing families,
Anari--Oveis Gharan--Saberi--Srivastava give subexponential constructions
via largest-root approximation \cite{anari2017}.
Jourdan--Macgregor--Sun give an algorithm that is quasi-polynomial in
the number of vectors in a specified regime, with exponential dependence
on dimension \cite{jourdan2023}.
They also give polynomial-time algorithms for a constant-factor paving
in a dense equal-norm regime \cite{jourdanMacgregorSun2024}.
These results have different guarantees and computational models from
the local walk studied here.

Hardness results also depend on the requested constant.
Spielman--Zhang prove hardness of distinguishing zero-discrepancy
instances from instances with discrepancy bounded below at a constant
multiple of the square-root scale \cite{spielmanZhang2022};
Jourdan--Macgregor--Sun prove a related small-constant search hardness
result. Such statements concern finer guarantees than our coarse
constant and should be distinguished from our real-arithmetic model.
The contribution developed below is the local source-driven mechanism
and its complete finite implementation in that model.

Our companion paper \cite{kathuria2026ms} establishes the square and
rectangular Matrix Spencer bounds using the same free-probability
variational approach.
Upcoming work \cite{kathuria2026higherRank}, \emph{A Walk From
Free Probability to Matrix Discrepancy III: Higher Rank Kadison--Singer
and Spectrally Thin Trees}, will extend this program to higher-rank
Kadison--Singer problems and spectrally thin trees.
Lean formalizations of the main discrepancy theorems proved here have
also been completed; the formal proofs will be released shortly.

\paragraph{Reading order.}
Section~\ref{sec:setup} defines the potential and gives the algorithm.
Sections~\ref{sec:potential}--\ref{sec:sm:curvature} prove the
outward-step/curvature alternative. Section~\ref{sec:smallstep}
turns it into a finite deterministic signing walk. The SDP and its numerical
use follow in Sections~\ref{sec:sdp}--\ref{sec:walk}.
The appendices supply standard analytic facts, quantitative derivative
bounds, arithmetic accounting, and the full marginal free-edge identity.

\section{The potential and the local walk}\label{sec:setup}
\label{sec:earlyalgorithm}

\subsection{State, source, and scales}
Throughout the proof $n$ is the number of vectors and $m$ their
ambient dimension. Assume $m>0$ and set
\begin{equation}\label{eq:normalization}
 \epsilon=\max_i\|A_i\|,\quad \delta=\sqrt\epsilon,\quad
 D=2m,\quad \theta=\delta/\sqrt D,\quad u=64.
\end{equation}
The actual maximum $\epsilon$ satisfies $0<\epsilon\le\varepsilon$.
The boundary margin, outward step, and smoothing width are
\begin{equation}\label{eq:geometry}
 \rho=\frac1{100n^2},\qquad a=\rho/16,\qquad \zeta=a/10.
\end{equation}
A state is $x\in[-1,1]^n$. A coordinate is \emph{live} when $|x_i|<1$
and \emph{frozen} at an endpoint. Define
\[
 H(x)=\sum_i x_iA_i,\qquad B_i=I_2\otimes A_i,\qquad
 J=\diag(I_m,-I_m),\qquad \widehat H(x)=I_2\otimes H(x).
\]
The largest eigenvalue of $J\widehat H(x)=\diag(H(x),-H(x))$ is
$\|H(x)\|$. Every frozen contribution remains in this center.

A density is a positive semidefinite matrix of trace one. Traces are
unnormalized, $X^*$ denotes conjugate transpose, and matrix norms are
operator norms unless specified otherwise. Hermitian matrix space is
real, with pairing $\langle X,Y\rangle=\Re\Tr(XY)$ and
$\|X\|_{\rm HS}=\|X\|_{\rm F}=(\Tr X^*X)^{1/2}$.
Adjoints of linear maps use these pairings. The range and kernel of a
map $T$ are $\operatorname{ran}T=\{Tv:v\in\operatorname{dom}T\}$ and
$\ker T=\{v:Tv=0\}$. The support of a positive semidefinite matrix is
its range, equal to the orthogonal complement of its kernel; a
supported inverse is taken on this space. The symbols $\Re X,\Im X$
denote entrywise real and imaginary parts. All logarithms are natural.

\begin{definition}[Profile and density potential]\label{def:potential}
For $-1\le z\le1$, put
\begin{equation}\label{eq:sourceprofile}
 c(z)=u\bigl(1-z^2+\sqrt{1+\zeta^2}-\sqrt{z^2+\zeta^2}\bigr),
 \qquad c_i=c(x_i),\quad d_i=c_i/u.
\end{equation}
For a $D\times D$ matrix $Y$, define the source map
$\Omega_x(Y)=\sum_i c_i\Tr(B_iY)B_i$.
The unsquared fidelity of positive matrices is
$\Fid(S,M)=\Tr\sqrt{M^{1/2}SM^{1/2}}$. Set
\begin{equation}\label{eq:potential}
\begin{split}
 L(x,S)&=\Tr(J\widehat H(x)S)+2\Fid(S,\Omega_x(S))
                                     +2\theta\Tr\sqrt S,\\
 \mathscr F(x)&=\max_{S\succeq0,\,\Tr S=1}L(x,S).
\end{split}
\end{equation}
The density is a full $D\times D$ matrix.
\end{definition}
The first term tests the two spectral signs together. The second
measures the source that remains at $x$, and the last smooths the
density optimization. At a vertex the source vanishes. The source is
always determined by $x$; there is no separately maintained covariance.
The profile vanishes at the faces, is positive on every live
coordinate, and satisfies $c''\le-2u$. These properties connect the
potential to discrepancy, permit a change of coefficient scale in
the analysis, and supply negative curvature. Lemmas~\ref{lem:transport} and~\ref{lem:floor} prove
that $\mathscr F$ is smooth on each open face, dominates $\|H(x)\|$,
and starts at $\mathscr F(0)\le34\delta$.

\subsection{The algorithm}
A \emph{prepared} state has had every live coordinate within $\rho$
of an endpoint rounded to that endpoint. Define
$\operatorname{sgn}_+(z)=1$ for $z\ge0$ and $-1$ otherwise, and let
$\Call{Prepare}{x}$ carry out these replacements in label order.
Let $\Call{Value}{y,\nu}$ be the value of the SDP in
\eqref{eq:O7} reported by Model~\ref{model:solver}; it satisfies
$|\Call{Value}{y,\nu}-\mathscr F(y)|\le\nu$.
The exact spectral primitive is denoted $\Call{EVD}{Q}=(\Lambda,V)$,
where $V$ is orthogonal and $V^TQV=\Lambda$ is diagonal.

The numerical scales have one purpose each: $h$ is a common movement
length, $t$ is a spacing for estimating curvature, and the two $\nu$'s
are value errors. Let $M>1$ be the explicit derivative bound in
\eqref{eq:regularitybudget}, computed once from the input scales.
Choose
\begin{equation}\label{eq:algorithmsetup}
\begin{gathered}
 \beta=\delta/(100n),\qquad
 t=\min\{a,\sqrt{\beta/(128nM)}\},\qquad
 h=\min\{a,\sqrt{\beta/(4M)}\},\\
 \nu_{\rm H}=\beta t^2/(128n),\qquad
 \nu_{\rm loc}=\beta h^2/8,\qquad T=\lceil16n/h^2\rceil.
\end{gathered}
\end{equation}
Appendix~\ref{sec:regularity} derives $M$, and
Appendix~\ref{sec:runtime} proves that these choices give polynomial
work. The reader can first regard $h$ as a common sufficiently small
step and $\nu_{\rm H},\nu_{\rm loc}$ as its specified error allowances.

For completeness, the curvature routine is as follows. List the live
labels $I=\{i_1<\cdots<i_k\}$ and write
$f_y=\Call{Value}{y,\nu_{\rm H}}$. With $e_i$ the $i$th standard vector
in the original coefficient space, form
\begin{equation}\label{eq:algorithmstencils}
\begin{aligned}
 \widetilde{\mathcal H}_{pp}
 &=\frac{f_{x+te_{i_p}}-2f_x+f_{x-te_{i_p}}}{t^2},\\
 \widetilde{\mathcal H}_{pq}
 &=\frac{f_{x+te_{i_p}+te_{i_q}}-f_{x+te_{i_p}-te_{i_q}}
 -f_{x-te_{i_p}+te_{i_q}}+f_{x-te_{i_p}-te_{i_q}}}{4t^2}\quad(p<q),\\
 \widetilde{\mathcal H}_{qp}&=\widetilde{\mathcal H}_{pq},\qquad
 \widetilde K_{pq}=\tfrac12\widetilde{\mathcal H}_{pq}.
\end{aligned}
\end{equation}
The routine $\Call{Curvature}{x}$ returns $(I,\widetilde K)$, an
approximation to half the Hessian of the optimized potential in the
original live coordinates.

The other routine constructs the subspace for movement. For
$z=(x_{i_1},\ldots,x_{i_k})$, let $\Call{TangentFrame}{z}$ return a
matrix $Q_x$ with orthonormal columns spanning
$\mathcal T_x=\{g\in\R^k:z^Tg=0\}$. It is computed as follows.
If $z=0$, take $Q_x=I_k$. Otherwise put
$\widehat z=z/\|z\|_2$, $w=\widehat z+\operatorname{sgn}_+(\widehat z_1)e_1$,
and take columns $2,\ldots,k$ of
\[
 R_x=I_k-2ww^T/\|w\|_2^2.
\]
This reflection sends $e_1$ to a signed copy of $\widehat z$, so its
other columns span $\mathcal T_x$. Moreover $\|w\|_2^2\ge2$.
Thus constructing the frame requires only the scalar primitives in
Model~\ref{model:solver}. In the branch where the frame is used, the
curvature theorem guarantees that it has at least one column.

\begin{algorithm}[H]
\caption{The deterministic local signing walk}\label{alg:trial}
\small
\begin{algorithmic}[1]
\Require Atoms $A_i$ with $m>0$, scales \eqref{eq:normalization}--\eqref{eq:algorithmsetup}, and the defined routines
\State $x\gets\Call{Prepare}{0}$; $j\gets0$
\While{some coordinate is live and $j<T$}
 \State $I\gets\{i:|x_i|<1\}$; $\widehat f_0\gets\Call{Value}{x,\nu_{\rm loc}}$
 \State For each $i\in I$, query $\widehat f_i\gets\Call{Value}{x+a\operatorname{sgn}_+(x_i)e_i,\nu_{\rm loc}}$
 \State Let $i_*$ be the least $i\in I$ with $\widehat f_i-\widehat f_0\le2\nu_{\rm loc}$, if one exists
 \If{$i_*$ exists}
  \State $x\gets x+a\operatorname{sgn}_+(x_{i_*})e_{i_*}$
 \Else
  \State $(\{i_1<\cdots<i_k\},\widetilde K)\gets\Call{Curvature}{x}$
  \State $Q_x\gets\Call{TangentFrame}{(x_{i_1},\ldots,x_{i_k})}$
  \State $(\Lambda,V)\gets\Call{EVD}{Q_x^T\widetilde KQ_x}$; choose the first $p_*$ minimizing $\Lambda_{pp}$
  \State Set $g_{i_p}=(Q_xV_{\cdot p_*})_p$ on live labels and $g_i=0$ on frozen labels
  \State $q_+\gets\Call{Value}{x+hg,\nu_{\rm loc}}$; $q_-\gets\Call{Value}{x-hg,\nu_{\rm loc}}$
  \State Set $x\gets x+hg$ if $q_+\le q_-$, and $x\gets x-hg$ otherwise
 \EndIf
 \State $x\gets\Call{Prepare}{x}$; $j\gets j+1$
\EndWhile
\State \Return $x$
\end{algorithmic}
\end{algorithm}

Every curvature movement uses a physical unit vector $g$ with $x^Tg=0$.
Consequently either sign increases $\|x\|_2^2$ by exactly $h^2$.
Curvature bounds the average spectral cost of the two candidates,
and comparing their reports converts this average bound into a bound
for the chosen movement. Outward movements also increase $\|x\|_2^2$,
and rounding preserves this progress. The analysis proves that the
returned array is always a full signing with the required norm bound.

\subsection{Why a covariance map belongs in the potential}
\label{sec:potentialmotivation}
For a Hermitian matrix $Y$, the source-free potential is
\[
 f_\theta(Y)=\max_{S\succeq0,\,\Tr S=1}
                  \{\Tr(YS)+2\theta\Tr\sqrt S\}.
\]
Applied to $J\widehat H(x)$ it is convex in $x$, because it is the
maximum of affine functions of $x$. Its density regularizer smooths
the spectral maximum, but its values at two symmetric candidates have
average at least its current value. A decreasing source creates a
competing contribution.

The ordinary remaining variance
$V(x)=\sum_i(1-x_i^2)A_i^2$ is a natural first candidate. It records
only the value at the identity of the map
$Y\mapsto\sum_i(1-x_i^2)A_iYA_i$. For a rank-one atom,
\[
 A_iYA_i=(v_i^*Yv_i)A_i,\qquad A_i^2=\|v_i\|^2A_i.
\]
The map retains how a test matrix weights each atom's direction.
Our shared source retains this information across both spectral signs:
for a density $S$ and a positive transport $Z$,
\[
 \Tr(S\Omega_x(Z))=\sum_i c_i\Tr(SB_i)\Tr(ZB_i).
\]
Section~\ref{sec:potential} derives the transport. The two factors in
each summand will measure the atom through the optimizing density and
transport. They appear in both the outward certificate and the bound
on reoptimization cost. This quantitative connection is what our proof
needs; the argument does not exclude other potentials built from
ordinary remaining variance.

The free model supplies a reason for this particular connection.
Its operator-valued covariance is the entire map $\Omega_x$, and
Lehner's formula represents its upper edge through
$Z^{-1}+\Omega_x(Z)$, optimized over $Z\succ0$.
Appendix~\ref{sec:freeinterpretation} proves the regularized identity
\[
 \mathscr F(x)=\inf_{Z\succ0}
 f_\theta\bigl(J\widehat H(x)+Z^{-1}+\Omega_x(Z)\bigr).
\]
Thus the auxiliary free operator is encoded by a finite-dimensional
variational problem. Freeness supplies its higher-order structure and
spectral formula; it is not implied by a covariance map alone.
The regularizer acts on the matrix marginal of a test state, and the
optimizer changes as the center and source evolve. Its response, rather
than the source norm in isolation, is the object analyzed below.

\subsection{The scalar example and the choice of profile}
Take $m=1$, write $A_i=a_i\ge0$ with $\sum_i a_i=1$, and put
$t=\sum_i x_i a_i$ and $V=\sum_i c_i a_i^2$.
Now $B_i=a_iI_2$ and $\Omega_x(S)=VI_2$ for every density.
Pinching to the basis of the signed center gives
\begin{equation}\label{eq:scalar}
 \mathscr F(x)=\max_{0\le s\le1}
 \left\{(2s-1)t+2(\sqrt V+\theta)
                         (\sqrt s+\sqrt{1-s})\right\}.
\end{equation}
Indeed, pinching preserves the center term and increases the trace
square root by concavity; the fidelity here is $\sqrt V\Tr\sqrt S$.
The same $V$ protects the positive and negative edges. Moving a
coordinate changes $t$, changes $V$, and changes the maximizing mass
$s$. This last change is the scalar version of the density response.
For $V>0$, the transport associated with the optimizing density
$S=\diag(s,1-s)$ is
\[
 Z=\frac1{\sqrt V}\diag(\sqrt s,\sqrt{1-s}),\qquad
 q_i=\Tr(ZB_i)=\frac{a_i(\sqrt s+\sqrt{1-s})}{\sqrt V}.
\]
It solves $Z(VI_2)Z=S$. The general transport formula will be proved
in Section~\ref{sec:potential}. Here it makes the outward comparison
visible: holding $Z$ fixed, a step of length $a$ changes the linear
density matrix by $aa_i(s_iJ-\ell_iq_iI_2)$, where $s_i$ is the
outward sign and $\ell_i=[c(x_i)-c(x_i+as_i)]/a$. Thus
$\ell_iq_i\ge1$ makes both diagonal entries nonpositive and certifies
the step. The certificate already measures source removal through
the current optimizer, even in this scalar example.

For $r\ge0$ and $r+a\le1$, the outward secant of the chosen profile is
\[
 \ell(r)=\frac{c(r)-c(r+a)}a
 \ge u(2r+a+9/10).
\]
The square-root difference is at least $a-\zeta=9a/10$.
The quadratic term supplies uniform negative second derivative; the
smoothed absolute-value term makes the finite secant substantial even
at zero. In contrast, the profile $1-z^2$ has secant $a$ there.
These are different jobs, both of which are needed by the argument.

\begin{figure}[t]
\centering
\includegraphics[width=\linewidth]{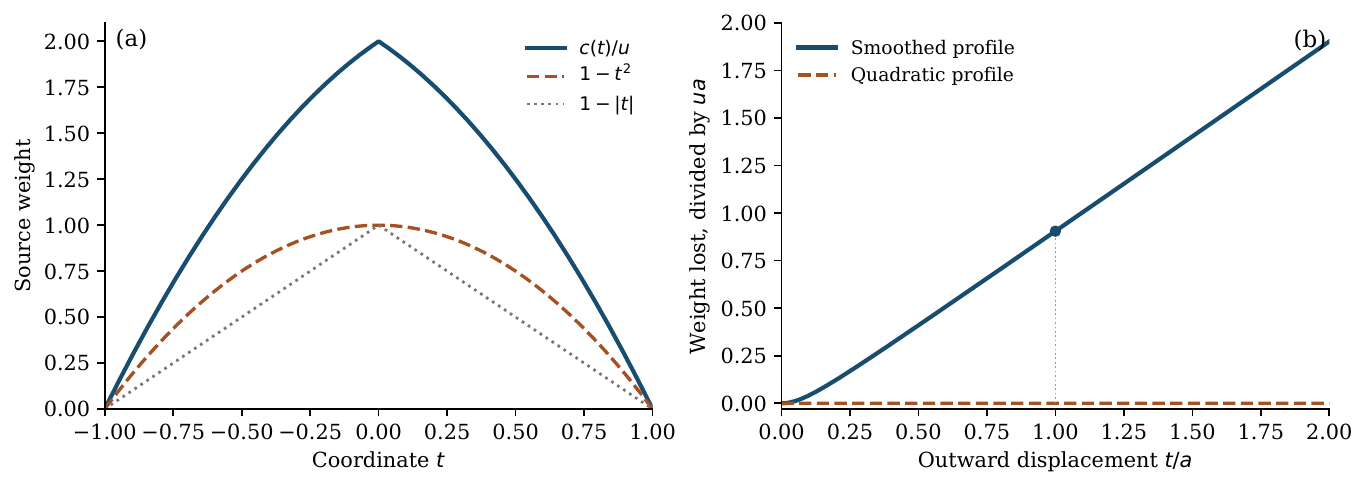}
\caption{The profile and its outward decrease, using $n=8$ and the
actual scales $a=1/(1600n^2)$, $\zeta=a/10$. The left panel separates
the quadratic and tent contributions. The right panel plots
$[c(0)-c(t)]/(ua)$ against $t/a$, compared with the quadratic profile's
loss. Although $c'(0)=0$, a step of length $a$ sees a substantial secant.}\label{fig:profile}
\end{figure}

\section{A density potential with a shared source}\label{sec:potential}

The source and density in Section~\ref{sec:setup} have been defined on
the whole cube. We now establish the properties needed to differentiate
their optimized value. The key step is a transport representation of
fidelity: it replaces its square root by linear terms in the density
and a minimization over a positive matrix. The same representation
will later certify outward steps and expose the optimizer's response.

\begin{lemma}[Parseval scales]\label{lem:scales}
If $m>0$, then $n\ge m\ge1$, $m/n\le\epsilon\le1$,
$\delta^{-1}\le\sqrt n$, and $\theta^{-1}\le\sqrt{2n}$.
If $m=0$, every signing has discrepancy zero.
\end{lemma}
\begin{proof}
$0\preceq A_i\preceq I$ and rank one give
$\Tr A_i=\norm{A_i}\le1$. Taking traces of Parseval gives
$m=\sum_i\norm{A_i}\le n\epsilon$. Combining this with $\epsilon\le1$ gives
the stated bounds on $\delta$ and $\theta$.
\end{proof}

\begin{lemma}[Properties of the source]\label{lem:sm:weights}
The scalar function $c(t)$ in \eqref{eq:sourceprofile} is smooth on the real
line, positive for $|t|<1$, and zero at $t=\pm1$. On $[-1,1]$ it satisfies
\begin{equation}\label{eq:sm:weightbounds}
\begin{gathered}
 u(1-t^2)\le c(t)\le2u,\\
 c'(t)=-u\left(2t+\frac{t}{\sqrt{t^2+\zeta^2}}\right),\qquad
 c''(t)=-2u-\frac{u\zeta^2}{(t^2+\zeta^2)^{3/2}}\le-2u.
\end{gathered}
\end{equation}
\end{lemma}
\begin{proof}
The derivatives follow directly and $\zeta>0$. Since $c$ is even and
decreasing on $[0,1]$, its maximum is
$c(0)=u(1+\sqrt{1+\zeta^2}-\zeta)\le2u$; the last inequality uses
$\sqrt{1+\zeta^2}\le1+\zeta$. Its additional square-root term is
nonnegative on the cube, giving the lower bound and endpoint statements.

\end{proof}

\begin{lemma}[Complete positivity and sign coupling]\label{lem:kraus}
Let $\sigma_0=I_2$ and let $\sigma_1,\sigma_2,\sigma_3$ be the Pauli
matrices, defined explicitly by
\[
 \sigma_1=\begin{pmatrix}0&1\\1&0\end{pmatrix},\quad
 \sigma_2=\begin{pmatrix}0&-\mathrm i\\\mathrm i&0\end{pmatrix},\quad
 \sigma_3=\begin{pmatrix}1&0\\0&-1\end{pmatrix}.
\]
Then
\begin{equation}\label{eq:kraus}
 \Omega_x(X)=\sum_{i=1}^n\sum_{\alpha=0}^3
 \frac{c_i}{2}(\sigma_\alpha\otimes A_i)X
                         (\sigma_\alpha\otimes A_i).
\end{equation}
Moreover $\Omega_x(I_D)\preceq4u\epsilon I_D$.
\end{lemma}
\begin{proof}
For $Y\in\C^{2\times2}$, direct multiplication gives
$\tfrac12\sum_{\alpha=0}^3\sigma_\alpha Y\sigma_\alpha=\Tr(Y)I_2$.
For $W\in\C^{m\times m}$, rank one gives
$A_iWA_i=\Tr(A_iW)A_i$.
Applying these identities to each sign block proves
\eqref{eq:kraus}, including off-diagonal inputs. The factors are
Hermitian and $c_i\ge0$, so this is a Kraus representation. It gives
complete positivity by the definition in Section~\ref{sec:preliminaries}.
Finally,
$\Omega_x(I_D)=2I_2\otimes\sum_i c_iA_i^2$ and
$A_i^2\preceq\epsilon A_i$, and $c_i\le2u$, which give the norm bound.
\end{proof}

The Kraus representation proves complete positivity on all matrix
blocks, including off-diagonal sign blocks. This is why the Pauli
identity is useful even for real input vectors. The density
optimization couples the two signs through one source and one trace
constraint.

\begin{lemma}[Continuity up to every face]\label{lem:continuity}
The potential $\mathscr F$ is continuous on the closed
cube. In particular it attains a minimum there, including when the
source support changes at the boundary.
\end{lemma}
\begin{proof}
The density set is compact. The entries of $\Omega_x(S)$ are smooth
in the real coordinates of $x$ and linear in $S$, and the positive square root is
continuous on the positive semidefinite cone. Thus $L(x,S)$ is jointly
continuous on a compact set and is uniformly continuous there. For
any two coefficient vectors,
\[
 |\mathscr F(x)-\mathscr F(y)|
 \le\max_{S\succeq0,\,\Tr S=1}|L(x,S)-L(y,S)|.
\]
The right side tends to zero as $y\to x$. Thus continuity also holds when the source loses rank at a face.
\end{proof}

\subsection{Transport and the optimizing density}

The transport identity turns a nonlinear fidelity term into an
optimization of linear density terms. We will use it in two ways:
fixing the old transport gives a finite comparison for a local
outward increment, while differentiating its stationarity equation gives the exact
curvature response. The source support is compressed only for the
transport; the density maximization remains full-dimensional. We call a density
\emph{faithful} when it is positive definite on the stated space.
The transport identity used below is the fidelity variational formula
of \cite[Theorem~3.19]{watrous2018}; we prove it here, including the
support and attainment details needed by the walk.
\begin{lemma}[Transport, faithfulness, and smoothness]\label{lem:transport}
On a fixed open live face, the maximum in \eqref{eq:potential} has a
unique positive definite optimizer. It commutes with $J$, as does the
positive transport on the live source support. The optimizing density
and potential are smooth in $x$ on that face. Here an open live face fixes
some coordinates at endpoints and lets the remaining coordinates vary
in $(-1,1)$.
\end{lemma}
\begin{proof}
\emph{The supported transport.}
Let $\mathcal K=\C^2\otimes\operatorname{span}\{v_i:i\text{ live}\}$
be the current source support, and let $S_0$ be the compression of a
full density to $\mathcal K$. A supported matrix is extended by zero
when used in a full-space trace; products with supported matrices use
the compressed atoms. Inverses of supported positive matrices are taken
only on $\mathcal K$. We also use $\Omega_x$ for the induced map
on matrices on $\mathcal K$, obtained by compressing the atoms in its
finite sum. Its zero extension agrees with the original full-space
map. The source depends only on this compression. At a faithful density,
$\Tr(B_iS)>0$ for every nonzero live atom. Their ranges span
$\mathcal K$, so their positively weighted sum $\Omega_x(S)$ is
positive definite there. If $\mathcal K=\{0\}$, the fidelity is zero
and every transport term below is absent. For positive \(S_0\) and
\(M=\Omega_x(S)\succ0\) on \(\mathcal K\),
\begin{equation}\label{eq:7}
 2\operatorname{Fid}(S,M)
 =\min_{Z\succ0\ {\rm on}\ \mathcal K}
   \{\operatorname{Tr}(S_0Z^{-1})+\operatorname{Tr}(MZ)\},
 \qquad ZMZ=S_0 .
\end{equation}
Compression is legitimate because
\(M^{1/2}SM^{1/2}=M^{1/2}S_0M^{1/2}\).
To prove the variational identity, set
\(Y=S_0^{1/2}Z^{-1}S_0^{1/2}\) and
\(Q_0=S_0^{1/2}MS_0^{1/2}\). The objective becomes
\(\operatorname{Tr}Y+\operatorname{Tr}(Q_0Y^{-1})\), and
\[
 \|Y^{1/2}-Y^{-1/2}Q_0^{1/2}\|_{\rm HS}^2
 =\operatorname{Tr}Y+\operatorname{Tr}(Q_0Y^{-1})
       -2\operatorname{Tr}\sqrt{Q_0}.
\]
The squared norm vanishes precisely at $Y=Q_0^{1/2}$, proving
attainment and uniqueness. Returning to the
original variables gives $ZMZ=S_0$. In particular the minimizer is
\[
 Z=S_0^{1/2}(S_0^{1/2}MS_0^{1/2})^{-1/2}S_0^{1/2},
\]
which depends smoothly on positive $S_0,M$. The equality of the two
square-root traces used here follows because $CC^*$ and $C^*C$ have
the same nonzero eigenvalues, with $C=M^{1/2}S_0^{1/2}$.
The identity proves joint
concavity in $(S_0,M)$, since its right side is an infimum of linear
functions of that pair. Singular density inputs follow by continuity; the
old positive test in \eqref{eq:7} remains an upper bound if a later source has
smaller support. More precisely, on any fixed finite-dimensional space,
for arbitrary $S,M\succeq0$ the same formula holds with an infimum
instead of a minimum. To see this, apply the positive formula to
$S+\eta I,M+\eta I$ for $\eta>0$. Every fixed positive test remains an
upper bound as $\eta\downarrow0$, proving the lower bound for the
infimum. Conversely, at the minimizing test for the perturbed pair,
the unperturbed objective is no larger than
$2\Fid(S+\eta I,M+\eta I)$. Continuity proves the reverse bound.
The infimum can fail to be attained when one input is singular.

\emph{The maximizing density.}
For fixed \(x\), the fidelity term is concave in the full \(S\):
by \eqref{eq:7} it is an infimum of linear functions of \(S\). The Tsallis term
is strictly concave. Consequently the maximum in \eqref{eq:potential} is unique.
It is attained at $S\succ0$. To prove this, suppose the optimizer
has a zero eigenvalue and mix it with $I_D/D$ with weight $t>0$.
Concavity gives a lower bound $(1-t)F(S)+tF(I_D/D)$
for the part $F$ consisting of the center and fidelity, so its loss
is at most a fixed multiple of $t$. The density commutes with $I_D$;
its zero eigenvalues under this mixture become $t/D$, and contribute
a positive multiple of $\sqrt t$ to $2\theta\Tr\sqrt S$. The other
eigenvalues change their square-root contribution by $O(t)$.
This makes the contradiction quantitative as $t\downarrow0$.

The objective is invariant under \(S\mapsto JSJ\). For the source this
uses \(\operatorname{Tr}(B_iJSJ)=\operatorname{Tr}(B_iS)\) and the fact
that every output commutes with \(J\). Concavity or uniqueness gives
\begin{equation}\label{eq:8}
 S=\operatorname{diag}(S_+,S_-).
\end{equation}
Its source, compression, and unique actual transport \(Z\) also commute
with \(J\). In particular \(S_\pm\succ0\); their sum has trace one.

\emph{Stationarity.}
The regularizer and trace constraint remain on the full space.
Define the density gradient by
$D_S L(x,S)[X]=\langle\nabla_S L(x,S),X\rangle$ for full Hermitian
variations $X$, using the real trace pairing. At the actual optimizer,
the full density stationarity equation, with scalar Lagrange multiplier
$\lambda\in\R$ for $\Tr S=1$, is written below. A tilde denotes
extension from $\mathcal K$ by zero; in particular
$\widetilde{Z^{-1}}$ is a full $D\times D$ matrix.
\begin{equation}\label{eq:9}
 J\widehat H(x)+\widetilde{Z^{-1}}
      +\Omega_x(Z)+\theta S^{-1/2}=\lambda I_{2m}.
\end{equation}
For fixed $x$, define
\[
 g_x(S,Z)=\Tr(S_0Z^{-1})+\Tr(\Omega_x(S_0)Z).
\]
For a supported Hermitian variation $U$, differentiating the inverse
by $(Z+tU)^{-1}{}'|_{t=0}=-Z^{-1}UZ^{-1}$ gives
\[
 D_Zg_x(S,Z)[U]
 =\Tr\bigl((\Omega_x(S_0)-Z^{-1}S_0Z^{-1})U\bigr)=0
\]
at the transport optimizer. Along the smooth minimizing transport,
the chain-rule term involving its derivative therefore vanishes.
For a full Hermitian density variation $X$, the remaining derivative is
\[
 D_S\bigl(2\Fid(S,\Omega_x(S))\bigr)[X]
 =\Tr\bigl((\widetilde{Z^{-1}}+\Omega_x(Z))X\bigr).
\]
This uses self-adjointness of the source, which follows directly from
\[\Tr(Y\Omega_x(X))=\sum_i c_i\Tr(B_iX)\Tr(B_iY).\]
The square-root regularizer has derivative
$\theta\Tr(S^{-1/2}X)$, derived just below.
Since the maximizing density is positive definite, both signs of every
sufficiently small trace-zero variation are feasible. Its gradient is
therefore orthogonal to all trace-zero Hermitian matrices, and hence
is a scalar multiple of the identity. This proves \eqref{eq:9}.
Conversely, concavity gives
$L(x,T)\le L(x,S)+\langle\nabla_S L(x,S),T-S\rangle$
for every density $T$, including singular $T$ by continuity.
Thus \eqref{eq:9} is sufficient as well as necessary for the maximum.
At the positive definite optimum the cone constraint is inactive.

\emph{Smooth dependence on the coefficients.}
The function $L(x,S)$ is jointly smooth when $x$
lies on a fixed open live face and $S\succ0$. One direct justification is as follows. The positive
square root and the transport in \eqref{eq:7} are smooth there. In an eigenbasis
of $S$, write $s_a>0$ for its eigenvalues and $X_{ab}$ for the entries
of a full Hermitian variation $X$. Denote the negative second variation
of the regularizer by $\mathcal A_\theta[X,X]$; explicitly,
\begin{equation}\label{eq:10}
 \mathcal A_\theta[X,X]
 =\theta\sum_{a,b}
  \frac{|X_{ab}|^2}
  {\sqrt{s_a}\sqrt{s_b}(\sqrt{s_a}+\sqrt{s_b})}>0
 \quad (X\ne0).
\end{equation}
For a direct derivation, set $Q=\sqrt S$ and differentiate $Q^2=S$.
The derivative $Y=D\sqrt S[X]$ solves $QY+YQ=X$, so in an eigenbasis
$Y_{ab}=X_{ab}/(\sqrt{s_a}+\sqrt{s_b})$. Taking traces gives
$D(2\theta\Tr\sqrt S)[X]=\theta\Tr(S^{-1/2}X)$.
Differentiating $Q^{-1}$ gives $-Q^{-1}YQ^{-1}$, which proves
\eqref{eq:10}, including coincident eigenvalues. It proves strict
concavity on the positive cone. For two distinct semidefinite inputs,
restrict to the range of their sum: every strict convex combination
is positive definite on that range, so the same negative second
derivative and endpoint continuity prove strict concavity there too.
Thus the density Hessian is strictly negative on
the full trace-zero tangent space. Its stationarity equations, with the trace
constraint, have an invertible derivative. Indeed, a vector $(X,b)$
in the kernel of this derivative satisfies
$D_S(\nabla_S L)(x,S)[X]=bI_D$ and $\Tr X=0$. Pairing with $X$ and using strict
negative definiteness forces $X=0$, then $b=0$. This is a square
finite-dimensional linear system, so injectivity implies invertibility.
The implicit function theorem (Theorem~\ref{thm:standardift}) now applies on the open positive cone
and gives a smooth local stationary density and multiplier. Concavity
and uniqueness identify this local solution with the actual optimizer.
The local solutions therefore join to give the asserted smooth optimizer
on each open face.
\end{proof}
\begin{remark}[Zero atoms]\label{rem:zero}
The source support is $\C^2\otimes\operatorname{span}\{v_i:i\text{ live}\}$.
Statements involving positive probes are restricted to nonzero live
atoms. The algorithm retains zero labels. Their local outward candidates
have exactly unchanged potential and are always accepted by the value
tests. Thus a curvature query is made only when every live atom is
nonzero; see Lemma~\ref{lem:localtest}. The response estimates are consequently applied only to positive probes.
\end{remark}

\subsection{The size of the potential and its optimizing density}

\begin{lemma}[Potential bounds and a density floor]\label{lem:floor}
For $x\in[-1,1]^n$,
\begin{equation}\label{eq:basicbounds}
\begin{split}
 \mathscr F(x)&\le1+34\delta,\\
 \mathscr F(0)&\le34\delta,\\
 \norm{H(x)}&\le\mathscr F(x).
\end{split}
\end{equation}
Every optimizing density has an explicit positive lower eigenvalue
bound, independent of the least positive source eigenvalue.
\end{lemma}
\begin{proof}
Parseval gives $-I\preceq H(x)\preceq I$.
Testing the transport infimum at $Z=tI$, for $t>0$, gives
$2\Fid(S,M)\le t^{-1}\Tr S+t\Tr M$.
Minimizing over $t$ gives $\Fid(S,M)\le\sqrt{\Tr S\Tr M}$;
if either trace vanishes, that matrix and its fidelity are zero. Since $\Tr\Omega_x(S)=\Tr(S\Omega_x(I))\le256\epsilon$,
the source term is at most $32\delta$.
Also $\Tr\sqrt S\le\sqrt D$. This proves the upper bounds.
For the lower bound choose a norming unit eigenvector of $H(x)$ and
put its rank-one density in the appropriate sign block. Its center
value is $\norm{H(x)}$, and both remaining terms are
nonnegative.

At the optimizer, pair \eqref{eq:9} with $S$. The two transport
terms each give the fidelity, so
$\lambda=\mathscr F(x)-\theta\Tr\sqrt S\le\mathscr F(x)$.
The transport terms are positive semidefinite. Hence
$\theta S^{-1/2}\preceq40I$.
In particular, with $\mu=(\theta/40)^2$,
\begin{equation}\label{eq:densityfloor}
 S\succeq\mu I_D.
\end{equation}
Indeed the potential is at most $35$ and the center has
norm at most $1$, so the constant $40$ is a conservative upper bound.
\end{proof}

The uniform bound $1+34\delta$ supplies conditioning on the entire cube.
The discrepancy estimate instead starts from the much smaller
initial value. The walk starts with $\mathscr F(0)\le34\delta$.
Lemma~\ref{lem:ledger} controls the potential at every state,
and Theorem~\ref{thm:trial} obtains the $O(\delta)$ discrepancy guarantee
from this smaller initial value and the completion estimate.

\section{Local outward comparison}\label{sec:outward}

An outward move changes two terms in the potential: it moves the
center and removes source. Holding the current minimizing transport
fixed bounds this combined change from above. This supplies a
sufficient certificate for an entire finite step. More significantly,
a failed value test tells us that the certificate must fail, which
will constrain the response calculation.

Use the candidate length $a$ and smoothing width $\zeta=a/10$.
Suppose a nonzero live atom has coordinate margin $1-|x_i|\ge a$, and choose
$s_i\in\{-1,1\}$ with $s_i x_i=|x_i|$, choosing $s_i=1$ at zero. Write $r=|x_i|$ and define its
outward secant
\begin{equation}\label{eq:sm:secant}
 \ell_i=\frac{c(r)-c(r+a)}a.
\end{equation}
At the current optimizing density and transport put
$p_i=\operatorname{Tr}(SB_i)>0$ and $q_i=\operatorname{Tr}(ZB_i)>0$
for each nonzero live label.

\begin{lemma}[Secant bounds and a certified increment]\label{lem:sm:outward}
The secant obeys
\begin{equation}\label{eq:sm:secantbounds}
 \ell_i\ge |c'(x_i)|,\qquad
 \ell_i\ge u\left(2|x_i|+a+\frac9{10}\right)\ge\frac9{10}u.
\end{equation}
If $\ell_iq_i\ge1$, the full increment of length $a$ satisfies
\begin{equation}\label{eq:sm:outward}
 \mathscr F(x+a s_i e_i)\le\mathscr F(x).
\end{equation}
\end{lemma}
\begin{proof}
Concavity on $[0,1]$ gives $\ell_i\ge-c'(r)=|c'(x_i)|$.
For the second bound, write
\[
 \ell_i=u(2r+a)+\frac u a
 \left[\sqrt{(r+a)^2+\zeta^2}-\sqrt{r^2+\zeta^2}\right].
\]
The bracket is at least $(r+a)-(r+\zeta)=a-\zeta=9a/10$.
This proves both secant estimates uniformly in $r$, including $r=0$.

For the potential comparison hold the old optimizing transport fixed.
Its density objective is a global upper bound for the new fidelity
objective by \eqref{eq:7}, also when a source weight becomes zero.
At the old state it is linear in the density except for the concave
Tsallis term. Its density gradient at the old optimizer equals the
actual potential objective's gradient, by the chain-rule calculation
in Lemma~\ref{lem:transport}. Equation~\eqref{eq:9} and the trace-one
constraint therefore show that the old density globally maximizes this
fixed-transport objective. Its maximum is exactly $\mathscr F(x)$.
This proves the needed saddle property directly.
The increment changes its linear density matrix by
\[
 a(s_iJ-\ell_iq_iI)B_i\preceq0.
\]
Indeed $B_i\succeq0$ commutes with $J$ and both sign-block coefficients
are nonpositive. Thus the new potential is at most the old maximum.
\end{proof}

The secant is a finite-step quantity. Its uniform lower bound near
zero is the role of the smoothed tent in the profile. If every
candidate increases the potential, the contrapositive of
Lemma~\ref{lem:sm:outward} gives $\ell_iq_i<1$ on every nonzero
live label. In particular,
\[
 |c'(x_i)|q_i\le1,\qquad c_iq_i^2\le\frac{200}{81u},
\]
using $c_i\le2u$. These are the two bounds that enter the curvature
argument. They control the optimizing transport without computing
it in the algorithm.

\section{Differentiating the optimized potential}\label{sec:curvature}

The next calculation is valid for smooth positive live weights given
by scalar functions $c_i(x_i)$. In this section $c_i$ without an argument
means its value at the current $x_i$, and $c_i'(x_i),c_i''(x_i)$ are its
ordinary scalar derivatives. For our walk every function $c_i$ is the
same function $c$ in \eqref{eq:sourceprofile}. All coefficient sums
in this section run over live labels. At the current optimizing density
$S$ and supported transport $Z$, define $p_i=\Tr(SB_i)$ and
$q_i=\Tr(ZB_i)$ for every live label, extending $Z$ by zero for a
full-space trace. Both probes vanish when $A_i=0$.
We keep the full density variation and the supported transport variation
together until the response has been identified. The particular source
profile will enter only through its first two derivatives.

\subsection{Linearizing the density and transport equations}
\label{sec:linearresponse}

The direct second derivative of the source weights is negative.
That observation alone is insufficient: the maximizing density and
minimizing transport both move. Their defining equations determine
these changes, and a Schur-complement calculation separates the
direct contribution from the cost of reoptimization.

This is the stability viewpoint suggested by matrix Dyson equations
\cite{aek2019}. Near a regular spectral edge, controlling the
linearized equation requires attention to its poorly controlled
directions \cite[Section~4]{aeks2020}. Here we analyze the coupled
regularized system selected by the variational problem, uniformly in
the directions required by the walk. Its regularizer term
$\theta S^{-1/2}$ is matrix-valued, so it generally cannot be absorbed
into a scalar spectral parameter. All derivatives below are on a
fixed open live face.

Let $\pi$ compress a full Hermitian matrix to the source support
$\mathcal K$, and let $\iota$ extend a matrix on $\mathcal K$ by zero.
These maps are adjoint for the trace pairing. In particular
$S_0=\pi S$. The density variation will remain a full matrix, including
its off-support entries. Only the transport is restricted to
$\mathcal K$.

Assume there is a live nonzero atom, so $\mathcal K\ne\{0\}$. If the
source support is zero, all live atoms are zero and can move outward
without changing the center or source; there are then no transport
equations to differentiate.

For any positive definite self-adjoint linear operator $\mathcal A$
on a real matrix inner-product space, use the weighted norm
\[
 \|Y\|_{\mathcal A}^2=\langle Y,\mathcal A(Y)\rangle.
\]
In particular, $\|Y\|_{\mathcal A^{-1}}^2$
means $\langle Y,\mathcal A^{-1}(Y)\rangle$: the inverse is that of a
\emph{linear operator}, rather than multiplication by an inverse
matrix. We also use
$\mathcal A[X,X]=\langle X,\mathcal A(X)\rangle$ for its quadratic form.

Fix $h\in\R^n$ with $h_i=0$ on every frozen label and set $x(t)=x+th$.
Write $S(t),Z(t),\lambda(t)$ for the actual optimizing density,
transport, and density multiplier. Unadorned $S,Z,\lambda$ denote their
values at $t=0$. Put
\[
 X=\dot S(0),\qquad U=\dot Z(0),\qquad
 \dot\lambda=\dot\lambda(0),\qquad X_0=\pi X.
\]
The trace constraint gives $\Tr X=0$. Define $W(t)=Z(t)^{-1}$,
$W_0=W(0)=Z^{-1}$, and
$M=\Omega_x(S_0)$ on $\mathcal K$. We use the restricted source map
when its argument is a matrix on $\mathcal K$. Define the following
operators on Hermitian matrices on $\mathcal K$, with $Y,V$ in that space:
\begin{equation}\label{eq:16}
\begin{split}
 \mathcal L(Y)&=W_0YW_0-\Omega_x(Y),\\
 \mathcal J_Z(V)&=W_0VM+MVW_0,\\
 \dot\Omega(Y)&=\sum_i c_i'(x_i)h_i\Tr(B_iY)B_i,\\
 N_h&=\dot\Omega(S_0),\qquad V_h=\dot\Omega(Z),\qquad
 B_h=\sum_i h_iJB_i.
\end{split}
\end{equation}
Here $B_h,V_h,N_h$ are written as matrices on $\mathcal K$. Their
extensions give the corresponding full matrices. In particular $\iota B_h$ is the derivative of
$J\widehat H(x+th)$ at $t=0$. Frozen coordinates are held fixed.

We also use $\mathcal A_\theta$ for the self-adjoint operator whose
quadratic form is \eqref{eq:10}. Its action follows from the square-root
Sylvester equation: with $Q=\sqrt S$, solve $QY+YQ=X$ and set
\[
 \mathcal A_\theta X=\theta Q^{-1}YQ^{-1}.
\]
Thus $D(\theta S^{-1/2})[X]=-\mathcal A_\theta X$. This operator acts
on the full density space, including entries outside the source support.

\begin{lemma}[The linearized KKT system]\label{lem:coupledresponse}
The actual derivatives of the optimizer satisfy
\begin{equation}\label{eq:coupledresponse}
\begin{aligned}
 \mathcal A_\theta X+\iota\mathcal L(U)+\dot\lambda I_D
     &=\iota(B_h+V_h),\\
 \mathcal J_Z(U)&=\mathcal L(X_0)-N_h,\\
 \Tr X&=0.
\end{aligned}
\end{equation}
This system has a unique solution $(X,U,\dot\lambda)$.
\end{lemma}
\begin{proof}
The nonlinear equations being differentiated are
\begin{equation}\label{eq:nonlinearresponse}
\begin{aligned}
 J\widehat H(x)+\iota W_0+\iota\Omega_x(Z)
                        +\theta S^{-1/2}&=\lambda I_D,\\
 W_0S_0W_0&=\Omega_x(S_0),\\
 \Tr S&=1.
\end{aligned}
\end{equation}
The first is full density stationarity. The second is the transport
equation $Z\Omega_x(S_0)Z=S_0$, written in inverse coordinates.
At the positive full density and supported transport, these
stationarity and normalization equations are the KKT conditions of the
regularized saddle-point problem. Differentiating them on the fixed
open live face yields the linearized KKT system
\eqref{eq:coupledresponse}.

Since $\dot W(0)=-W_0UW_0$, differentiation of the first equation gives
\[
 \iota B_h-\iota(W_0UW_0)+\iota V_h+\iota\Omega_x(U)
               -\mathcal A_\theta X=\dot\lambda I_D.
\]
Grouping its transport terms gives the first line of
\eqref{eq:coupledresponse}. Both terms in
$\left.\frac{d}{dt}\Omega_{x(t)}(Z(t))\right|_{t=0}
=\dot\Omega(Z)+\Omega_x(U)$ are
needed: one changes the coefficient weights, and the other changes
the transport.

For the second nonlinear equation, the product rule gives
\[
 -W_0UW_0S_0W_0+W_0X_0W_0-W_0S_0W_0UW_0
                   =N_h+\Omega_x(X_0).
\]
Using $W_0S_0W_0=M$ leaves
$W_0X_0W_0-\Omega_x(X_0)=W_0UM+MUW_0+N_h$.
This is the second line of \eqref{eq:coupledresponse}. The last line
is differentiated normalization.

To prove uniqueness, first observe that $\mathcal L$ is self-adjoint
and $\mathcal J_Z$ is positive definite on the source space:
\[
 \langle U,\mathcal J_ZU\rangle
 =2\Tr(MUW_0U)
 =2\norm{M^{1/2}UW_0^{1/2}}_{\mathrm F}^2>0\quad(U\ne0).
\]
Eliminating $U$ as below leaves a positive definite operator on the
full trace-zero density tangent. It determines $X$ uniquely; the
second line then determines $U$, and the trace of the first determines
$\dot\lambda$.
\end{proof}

\subsection{Eliminating the transport and identifying the response}

Let $\mathcal T=\{X=X^*: \Tr X=0\}$ and let
$\mathcal P_0(Y)=Y-D^{-1}(\Tr Y)I_D$ be the orthogonal projection
onto this full tangent space. From the second line of
\eqref{eq:coupledresponse},
\begin{equation}\label{eq:transportelimination}
 U=\mathcal J_Z^{-1}(\mathcal L\pi X-N_h).
\end{equation}
Substitute this into the first line and project by $\mathcal P_0$.
The scalar multiplier disappears, giving
\begin{equation}\label{eq:responseoperator}
\begin{aligned}
 \mathcal C_\theta
 &=\left.\mathcal P_0\bigl(\mathcal A_\theta
            +\iota\mathcal L\mathcal J_Z^{-1}\mathcal L\pi\bigr)
                                                \right|_{\mathcal T},\\
 b_h&=\mathcal P_0\iota
             (B_h+V_h+\mathcal L\mathcal J_Z^{-1}N_h),\\
 \mathcal C_\theta X&=b_h,\qquad X=\mathcal C_\theta^{-1}b_h.
\end{aligned}
\end{equation}
The inverse is taken only on $\mathcal T$. For $X\in\mathcal T$,
\begin{equation}\label{eq:responsecoercivity}
 \langle X,\mathcal C_\theta X\rangle
 =\mathcal A_\theta[X,X]
       +\|\mathcal L\pi X\|_{\mathcal J_Z^{-1}}^2>0\quad(X\ne0).
\end{equation}
The regularizer therefore makes the full response operator invertible.
The numerical size of $\mathcal J_Z^{-1}$, however, may reflect small
source eigenvalues. Section~\ref{sec:sm:curvature} changes coordinates
and bounds the needed response quadratic form without estimating that
inverse separately.

For every candidate density variation $X\in\mathcal T$, write
$X_0=\pi X$. Inside the maximization below this compression varies
with $X$ over the full trace-zero tangent.

\begin{lemma}[Joint Hessian formula]\label{lem:hessian}
The actual optimized coefficient Hessian has the response formula
\begin{equation}\label{eq:optimizedresponse}
 \mathscr F''(x)[h,h]
 =\sum_i c_i''(x_i)p_iq_i h_i^2-\|N_h\|_{\mathcal J_Z^{-1}}^2
                   +\langle b_h,\mathcal C_\theta^{-1}b_h\rangle.
\end{equation}
Equivalently,
\begin{equation}\label{eq:17}
\begin{aligned}
 \mathscr F''(x)[h,h]
 &= \sum_i c_i''(x_i)p_iq_i h_i^2+\max_{X\in\mathcal T}\bigl\{
          -\mathcal A_\theta[X,X]\\
 &\hspace{7mm}+2\langle X_0,B_h+V_h\rangle
       -\|\mathcal L X_0-N_h\|_{\mathcal J_Z^{-1}}^2\bigr\}.
\end{aligned}
\end{equation}
The maximizer is the full optimizer derivative $X=\dot S(0)$.
The formula separates the direct negative source curvature from the
remaining response maximization.
\end{lemma}
\begin{proof}
Introduce the joint objective
\[
 \mathcal G(x,S,Z)=\Tr((J\widehat H(x))S)
 +\Tr(S_0Z^{-1})+\Tr(\Omega_x(S_0)Z)+2\theta\Tr\sqrt S.
\]
Along its actual optimizers, its first derivative through the density
is $\lambda\Tr\dot S=0$, and its first derivative through the
transport is zero. Thus the first envelope derivative is
\[
 \mathscr F'(x)[h]=\Tr(S_0B_h)+\Tr(\dot\Omega(S_0)Z).
\]
Differentiate this expression once more. The coefficient weights obey
$\ddot c_i=c_i''(x_i)h_i^2$, so
\begin{equation}\label{eq:second-envelope}
 \mathscr F''(x)[h,h]
 =\sum_i c_i''(x_i)p_iq_i h_i^2
             +\langle X_0,B_h+V_h\rangle+\langle U,N_h\rangle.
\end{equation}
Substituting \eqref{eq:transportelimination} and using self-adjointness
of $\mathcal L$ combines the last two terms into
$\langle X,b_h\rangle-\|N_h\|_{\mathcal J_Z^{-1}}^2$.
Equation~\eqref{eq:responseoperator} proves
\eqref{eq:optimizedresponse}.

For the equivalent quadratic form, positivity of $\mathcal C_\theta$
gives
\[
 \langle b_h,\mathcal C_\theta^{-1}b_h\rangle
 =\max_{X\in\mathcal T}
       \{2\langle b_h,X\rangle-\langle X,\mathcal C_\theta X\rangle\}.
\]
Indeed, writing $X_*=\mathcal C_\theta^{-1}b_h$, the expression
being maximized equals
\[
 \langle b_h,\mathcal C_\theta^{-1}b_h\rangle
 -\langle X-X_*,\mathcal C_\theta(X-X_*)\rangle.
\]
Positive definiteness on $\mathcal T$ proves both attainment and
uniqueness of this maximum.
Insert \eqref{eq:responseoperator}--\eqref{eq:responsecoercivity} and
complete the square in $\mathcal L X_0$ to obtain \eqref{eq:17}.
In particular the source change produces a shifted negative square;
its shift is the same $N_h$ that appeared in the differentiated
transport equation.

One can also read this from the joint second variation of
$\mathcal G$:
\begin{equation}\label{eq:18}
\begin{aligned}
 &\sum_i c_i''(x_i)p_iq_i h_i^2-\mathcal A_\theta[X,X]
       +2\langle X_0,B_h+V_h\rangle\\
 &\quad+\langle U,\mathcal J_ZU\rangle
       -2\langle U,\mathcal L X_0\rangle+2\langle U,N_h\rangle.
\end{aligned}
\end{equation}
Minimizing in $U$ gives \eqref{eq:transportelimination}; maximizing on
the full trace tangent gives \eqref{eq:responseoperator}. This is the
quadratic form of the same coupled linear system.
\end{proof}

The quadratic maximization in \eqref{eq:17} is invariant under
$X\mapsto JXJ$ and is concave in $X$. Averaging a variation with its
conjugate does not decrease its value, so the maximizer can be taken
block diagonal. We use this symmetry in the balanced calculation,
while retaining the full regularizer and the full trace constraint
until the explicit relaxation in \eqref{eq:sm:response}.

\section{A general curvature estimate for concave scalar weights}
\label{sec:sm:curvature}

The Hessian formula contains a favorable source-curvature term and an
optimized response that could, in principle, overwhelm it. We show
that the failed outward tests prevent this. The proof has two steps:
balance the transport equation to reveal a positive fixed point, then
average over pairs of directions satisfying its linear response
constraint. The second step preserves enough negative curvature
without a bound on an inverse spectral gap. We also impose one prescribed
linear condition on the coefficient direction. This will make the
physical movement orthogonal to the current coefficient vector.

We first treat general concave scalar source weights. At the state
under consideration, let the positive live weights be given by smooth
scalar functions $c_i(x_i)$ satisfying
\[
 c_i''(x_i)\le-2u.
\]
Assume at least one nonzero live label remains, and let $k\ge1$ be
the number of such labels. All coefficient vectors
in this section lie in $\R^k$, and coefficient-indexed matrices are
$k\times k$ matrices on these labels. Directions are extended by zero
to the original $n$ coordinates. Zero atoms and frozen labels are held
fixed. As in the preceding section, $c_i$ without an argument denotes
the value $c_i(x_i)$, and $p_i=\Tr(SB_i)$, $q_i=\Tr(ZB_i)$.
The full density potential is still \eqref{eq:potential}. Define
\begin{equation}\label{eq:sm:slopes}
 \alpha_i=-c_i'(x_i),\qquad z_i=\alpha_iq_i,\qquad
 Z_c=\operatorname{diag}(z_i),
\end{equation}
and, for $y\in\R^k$, normalize a coefficient direction by
\begin{equation}\label{eq:sm:hessian}
 D_d=\operatorname{diag}(c_i/u),\qquad h=D_d^{1/2}y,\qquad
 K=\tfrac12D_d^{1/2}\nabla^2\mathscr F(x)D_d^{1/2}.
\end{equation}

\subsection{Balanced variables and the joint density response}

Let $S_0$ be the compression of the optimizing density to
$\mathcal K$. In this subsection all supported matrix variables
commute with the restriction of $J$ to $\mathcal K$.
Congruence by $Z^{-1/2}$ gives the density and source the same
normalization. In these coordinates the transport equation becomes
a fixed-point equation, and the source variation becomes a linear
combination of balanced atoms. Define
\begin{equation}\label{eq:sm:balanced}
\begin{gathered}
 P=Z^{-1/2}S_0Z^{-1/2},\qquad D_i=\sqrt{c_i}\,Z^{1/2}B_iZ^{1/2},\\
 r_i=\Tr D_i=\sqrt{c_i}q_i,\qquad
 \mu_i=\Tr(PD_i)=\sqrt{c_i}p_i.
\end{gathered}
\end{equation}
Both $P$ and $D_i$ act on $\mathcal K$, and $r_i,\mu_i>0$ on the
nonzero live labels. Write $\boldsymbol\mu=(\mu_i)_i\in\R^k$.
Next define the coefficient-to-matrix map and its adjoint by
\[
 \mathcal D y=\sum_i y_iD_i,\qquad
 (\mathcal D^*Y)_i=\Tr(D_iY).
\]
The adjoint uses the Euclidean and Hilbert--Schmidt inner products.
Multiplication $Y\mapsto JY$ preserves the supported block-diagonal
Hermitian space. On this space and on $\R^k$, respectively, set
\[
 \Phi=\mathcal D\mathcal D^*,\qquad
 T=\mathcal D^*\mathcal D,\qquad T_J=\mathcal D^*J\mathcal D.
\]
Thus $\Phi$ is a linear map on matrices, while $T,T_J$ are real
$k\times k$ matrices. In this subsection inequalities between
self-adjoint linear maps use the Hilbert--Schmidt inner product:
$\Phi\succeq0$ means $\langle Y,\Phi(Y)\rangle\ge0$ for every $Y$.
Finally define three more coefficient matrices:
\[
 \Gamma_{ij}=\Re\Tr(PD_iD_j),\qquad
 G=\diag(r_i\mu_i),\qquad R=\diag(\mu_i/r_i).
\]

\begin{lemma}[Balanced-source identities]\label{lem:balanced}
The quantities just defined satisfy
\begin{equation}\label{eq:sm:identities}
\begin{gathered}
 P=\sum_i\mu_iD_i,\quad \Phi(P)=P,\quad
 0\preceq T\preceq I,\quad T\boldsymbol\mu=\boldsymbol\mu,\\
 T_{ii}\le r_i^2,\quad |(T_J)_{ii}|\le T_{ii},\quad
 \Gamma\succeq0,\quad \Gamma_{ii}\le r_i\mu_i,\quad
 TRT\preceq\Gamma,\quad T_JRT_J\preceq\Gamma.
\end{gathered}
\end{equation}

\end{lemma}
\begin{proof}
Each balanced atom $D_i$ has rank two. The first identities
follow by cyclicity and the transport equation:
\[
 P=Z^{1/2}\Omega_x(S_0)Z^{1/2}
   =\sum_i\sqrt{c_i}p_iD_i.
\]
Thus \(\Phi(P)=P\). The map \(\Phi(Y)=\sum_i\operatorname{Tr}(D_iY)D_i\)
is positivity preserving and self-adjoint positive semidefinite on
Hermitian matrix space. Its faithful fixed point makes it a contraction
in the norm $\|Y\|_P=\inf\{t\ge0:-tP\preceq Y\preceq tP\}$: if \(-tP\preceq Y\preceq tP\),
the same inequalities hold for \(\Phi(Y)\). Hence its spectral radius
is one. For a self-adjoint operator, the spectral radius equals the operator
norm; its Hilbert--Schmidt positivity therefore gives
\(0\preceq\Phi\preceq I\). The product-spectrum identity then gives
\(0\preceq T\preceq I\).

For that last identity, if $T v=\lambda v$ with $\lambda>0$, then
$\mathcal Dv\ne0$ and $\Phi(\mathcal Dv)=\lambda\mathcal Dv$.
Conversely a positive-eigenvalue vector for $\Phi$ maps under
$\mathcal D^*$ to one for $T$. The remaining eigenvalues are zero;
positivity of $T=\mathcal D^*\mathcal D$ follows directly from its
Gram form. Applying \(\mathcal D^*\) to the fixed point
gives \(T\boldsymbol\mu=\boldsymbol\mu\).

For real $v$, $v^T\Gamma v=\Tr(P(\mathcal Dv)^2)\ge0$, so
$\Gamma\succeq0$. This identity also explains the weighted matrix
square that will bound the density response.
The PSD inequality \(D_i^2\preceq r_iD_i\) gives the diagonal bound for
\(\Gamma\) and \(T_{ii}\le r_i^2\).
For every block-diagonal Hermitian \(Y\), Cauchy--Schwarz gives
\[
 |\operatorname{Tr}(D_iY)|^2
 \le r_i\operatorname{Tr}(D_iY^2).
\]
Indeed, apply Hilbert--Schmidt Cauchy--Schwarz to $D_i^{1/2}$ and
$YD_i^{1/2}$; their squared norms are $r_i$ and
$\Tr(D_iY^2)$, respectively.
Multiplying by \(\mu_i/r_i\) and summing proves
\begin{equation}\label{eq:22}
 \sum_i\frac{\mu_i}{r_i}|\operatorname{Tr}(D_iY)|^2
 \le\operatorname{Tr}(PY^2).
\end{equation}
Take \(Y=\mathcal Dv\) to obtain \(TRT\preceq\Gamma\).
Take \(Y=J\mathcal Dv\), which is Hermitian and satisfies
\(Y^2=(\mathcal Dv)^2\), to obtain \(T_JRT_J\preceq\Gamma\).
Finally \(|\operatorname{Tr}(D_iJD_i)|\le\operatorname{Tr}D_i^2\).
This establishes all the balanced-source inequalities.
\end{proof}

In the joint density/transport response formula, the direct second
derivative of the coefficient weights now contributes
\begin{equation}\label{eq:sm:direct}
 \frac12\sum_i c_i''(x_i)p_iq_i h_i^2
 \le-u\sum_i p_iq_i h_i^2=-y^TGy.
\end{equation}
The additional concavity in the smooth source improves this upper bound.
All first-derivative terms retain the same form. In detail, define
\begin{equation}\label{eq:sm:derivatives}
\begin{gathered}
 A_y=\frac{J\mathcal Dy-\mathcal DZ_cy}{\sqrt u},\qquad
 C_{ii}=\frac{\alpha_i p_i}{2\sqrt u},\qquad
 C=\frac1{2\sqrt u}RZ_c,\\
 Z^{1/2}\dot\Omega(S_0)Z^{1/2}=-2\mathcal DCy.
\end{gathered}
\end{equation}
For the last equality substitute $\dot c_i=-\alpha_i h_i$ and
$Z^{1/2}B_iZ^{1/2}=D_i/\sqrt{c_i}$. Thus the relation between the
source variation and the density force survives for arbitrary scalar
weights.

Write $\Lambda=I-\Phi$ and $\mathcal J_P(U)=PU+UP$.
To rewrite the response quadratic in these coordinates, for the
full density variation $X$, put $Y=Z^{-1/2}X_0Z^{-1/2}$. The definitions
of $\mathcal L$ and $\mathcal J_Z$ give, by direct congruence,
\[
 \mathcal L(X_0)=Z^{-1/2}\Lambda(Y)Z^{-1/2},\qquad
 \mathcal J_Z(Z^{1/2}UZ^{1/2})
       =Z^{-1/2}\mathcal J_P(U)Z^{-1/2}.
\]
Also $\langle X_0,B_h+V_h\rangle=\langle Y,A_y\rangle$,
by cyclicity of trace and \eqref{eq:sm:derivatives}. Thus the shifted
response norm becomes
$\|\Lambda Y+2\mathcal DCy\|_{\mathcal J_P^{-1}}^2$.

The joint response formula \eqref{eq:17}, \eqref{eq:sm:direct}, and dropping the
nonnegative full-density regularizer quadratic form give
\begin{equation}\label{eq:sm:response}
 y^TKy\le-y^TGy+
 \frac12\sup_Y\left\{2\langle A_y,Y\rangle
       -\|\Lambda Y+2\mathcal DCy\|_{\mathcal J_P^{-1}}^2\right\}.
\end{equation}
The supremum is over all block-diagonal Hermitian matrices on the source
support, an enlargement of the actual compressed density tangent space.
The inverse is a positive Sylvester inverse on this fixed support.
The remaining operator $\Lambda=I-\Phi$ has a kernel because
$\Phi(P)=P$. Estimating the response by $\Lambda^{-1}$ would therefore
lose its essential structure. Instead, we restrict coefficient
directions and an auxiliary vector together so that the linear
functional can be represented inside the range of $\Lambda$.

A pair $(\omega,y)$ satisfying the next equation will be called
\emph{legal}. The vector $y$ specifies a coefficient direction;
$\omega$ represents its compatible response. Legal pairs are an
analytic device for proving that a low-curvature direction exists.
The implemented algorithm obtains a direction from the Hessian.

\begin{lemma}[Legal response bound]\label{lem:sm:legal}
Every coefficient pair satisfying
\begin{equation}\label{eq:sm:legal}
 (I-T)\omega=(T_J-TZ_c)y
\end{equation}
obeys
\begin{equation}\label{eq:sm:legalbound}
\begin{split}
 y^TKy\le{}&-y^TGy-\frac1u\omega^TRZ_cy\\
 &+\frac1u\operatorname{Tr}
 P\bigl(J\mathcal Dy-\mathcal DZ_cy+\mathcal D\omega\bigr)^2.
\end{split}
\end{equation}
\end{lemma}
\begin{proof}
Put $\xi=\omega-Z_cy$. Then $(I-T)\xi=(T_J-Z_c)y$, and
\[
 U=(J\mathcal Dy+\mathcal D\xi)/\sqrt u
 \quad\hbox{satisfies}\quad \Lambda U=A_y.
\]
Since $\Lambda$ is self-adjoint and $\Lambda U=A_y$, the linear
term satisfies $\langle A_y,Y\rangle=\langle U,\Lambda Y\rangle$.
Set $b=2\mathcal DCy$ and temporarily allow $W=\Lambda Y$ to range
over all block-diagonal Hermitian matrices. This enlarges the domain,
so it gives an upper bound even if $\Lambda$ has a kernel. Direct
completion of the square gives
\[
\begin{split}
 2\langle U,W\rangle-\|W+b\|_{\mathcal J_P^{-1}}^2
 ={}&\langle U,\mathcal J_PU\rangle-2\langle U,b\rangle\\
 &-\|W+b-\mathcal J_PU\|_{\mathcal J_P^{-1}}^2.
\end{split}
\]
The enlarged maximum is attained at $W=\mathcal J_PU-b$.
Since $\langle U,\mathcal J_PU\rangle=2\Tr(PU^2)$, the factor
$1/2$ in \eqref{eq:sm:response} yields
\[
 y^TKy\le-y^TGy+\Tr(PU^2)-2\langle U,\mathcal DCy\rangle.
\]
Legality gives $T_Jy+T\xi=Z_cy+\xi=\omega$. By
$C=RZ_c/(2\sqrt u)$, the mixed term is exactly
$-\omega^TRZ_cy/u$. Substituting $\xi=\omega-Z_cy$ proves the claim.
\end{proof}

\subsection{Averaging over compatible directions}

A bound for one legal pair is not enough: we need a pair for which
the negative term dominates. The legal equation defines a linear
subspace of the product of two coefficient spaces. Its dimension is
large, and the failed-probe bounds ensure that most of the projection's
mass remains on movement coordinates. We now quantify this observation.

\begin{theorem}[Curvature from slope and probe bounds]\label{thm:sm:general}
Suppose the live set contains a nonzero atom and, for some $L>0$,
\begin{equation}\label{eq:sm:smallness}
 |z_i|\le1,\qquad r_i^2=c_iq_i^2\le L/u
\end{equation}
on every nonzero live label. For every $b\in\R^k$, there is a centered
random pair $(\omega,y)\in\R^k\oplus\R^k$ satisfying
\eqref{eq:sm:legal} and $b^Ty=0$ with probability one, for which
\begin{equation}\label{eq:sm:generalcurvature}
 \mathbb E[y^TKy]
 \le-\left(1-\frac{6L+7+\sqrt5+2\sqrt6}{u}\right)\sum_i c_iq_i^2.
\end{equation}
In particular $K$ has a negative quadratic value on $b^\perp$ whenever
$u>6L+7+\sqrt5+2\sqrt6$.
\end{theorem}
\begin{proof}
Put $h_*=\sum_i r_i^2$. We first record the quantitative lower bound
\begin{equation}\label{eq:source-mass}
 h_*\ge\Tr T\ge1.
\end{equation}
Indeed, $T\boldsymbol\mu=\boldsymbol\mu$ with
$\boldsymbol\mu\ne0$ gives an eigenvalue equal to one. Since $T\succeq0$,
its trace is at least one, and $T_{ii}\le r_i^2$ bounds that trace above.
This unit of source mass will pay for the additional linear constraint.
Whiten the legal variables by
$w=R^{1/2}\omega$, $v=R^{1/2}y$. Define
\[
 E=R^{1/2}TR^{-1/2},\qquad
 F=R^{1/2}(T_J-TZ_c)R^{-1/2}.
\]
The legal equation is $(I-E)w=Fv$. The balanced identities give
\[
 \|E\|_{\mathrm{HS}}^2\le h_*,\qquad
 \|R^{1/2}T_JR^{-1/2}\|_{\mathrm{HS}}^2\le h_*,\qquad
 \|F\|_{\mathrm{HS}}\le2\sqrt{h_*}.
\]
For example,
$\|E\|_{\mathrm{HS}}^2=\operatorname{Tr}(R^{-1}TRT)
\le\operatorname{Tr}(R^{-1}\Gamma)\le h_*$.
The estimate for $F$ follows by the triangle inequality and $\|Z_c\|\le1$.

Let $\Pi$ be the orthogonal projection in $\R^k\oplus\R^k$ onto the
kernel of the $(k+1)\times2k$ matrix
\[
 \begin{bmatrix}I-E&-F\\0&b^TR^{-1/2}\end{bmatrix}:
 (w,v)\longmapsto\bigl((I-E)w-Fv,\ b^TR^{-1/2}v\bigr).
\]
The first block row enforces legality; the last row enforces $b^Ty=0$.
Write the projection blocks as
$\Pi=\left(\begin{smallmatrix}\Pi_{ww}&\Pi_{wv}\\
\Pi_{vw}&\Pi_{vv}\end{smallmatrix}\right)$ in the variable order $(w,v)$. Put
$\beta_{\rm proj}=\|E\|_{\mathrm{HS}}^2+\|F\|_{\mathrm{HS}}^2\le5h_*$.
Its blocks satisfy
\begin{equation}\label{eq:sm:projectionblocks}
\begin{gathered}
 \operatorname{Tr}\Pi_{ww}\le\beta_{\rm proj},\qquad
 \operatorname{Tr}(I-\Pi_{vv})\le\beta_{\rm proj}+1\le6h_*,\\
 \|\Pi_{wv}\|_{\mathrm{HS}}\le\sqrt{\beta_{\rm proj}},\qquad
 \|I-\Pi_{vv}\|_{\mathrm{HS}}\le\sqrt{\beta_{\rm proj}+1}.
\end{gathered}
\end{equation}
To see the first bound in detail, the kernel relation gives
$[I\ 0]\Pi=[E\ F]\Pi$. Since $\Pi=\Pi^* =\Pi^2$,
\[
 \Tr\Pi_{ww}=\|[I\ 0]\Pi\|_{\rm HS}^2
 =\|[E\ F]\Pi\|_{\rm HS}^2
 \le\|[E\ F]\|_{\rm HS}^2=\beta_{\rm proj}.
\]
The constraint matrix has at most $k+1$ independent rows, so
$\Tr\Pi\ge k-1$. Consequently
$\Tr(I-\Pi_{vv})\le\Tr\Pi_{ww}+1\le\beta_{\rm proj}+1$.
Using $h_*\ge1$ gives the displayed bound by $6h_*$.
For the cross block,
$\|\Pi_{wv}\|_{\rm HS}^2=\Tr(\Pi_{wv}\Pi_{vw})
=\Tr(\Pi_{ww}-\Pi_{ww}^2)\le\Tr\Pi_{ww}\le\beta_{\rm proj}$.
The inequality $0\preceq I-\Pi_{vv}\preceq I$ proves the last bound.

Choose a centered vector $(w,v)$ with covariance $\Pi$, for example
$\Pi\zeta_0$, where $\zeta_0\in\{-1,1\}^{2k}$ has independent uniform entries. This vector is
supported on the constrained legal kernel. This finite average will
establish existence; the algorithm selects its direction by the
compressed eigendecomposition. We average the three terms in
\eqref{eq:sm:legalbound}.

First, $R^{-1/2}GR^{-1/2}=\operatorname{diag}(r_i^2)\preceq LI/u$,
so
\begin{equation}\label{eq:sm:negative}
 \mathbb E[y^TGy]
 \ge h_*-\frac L u\operatorname{Tr}(I-\Pi_{vv})
 \ge(1-6L/u)h_*.
\end{equation}

Second, the response-square quadratic form is positive semidefinite.
Since $\Pi\preceq I$, its expectation is at most its trace under
independent isotropic $w,v$. Commutation with $J$ and $|z_i|\le1$ give
\begin{equation}\label{eq:sm:square}
\begin{split}
 &\mathbb E\operatorname{Tr}
 P\bigl(J\mathcal Dy-\mathcal DZ_cy+\mathcal D\omega\bigr)^2\\
 &\quad\le\sum_i
 \frac{\operatorname{Tr}(P(J-z_iI)^2D_i^2)
                         +\operatorname{Tr}(PD_i^2)}{R_{ii}}
 \le5\operatorname{Tr}(R^{-1}\Gamma)\le5h_*.
\end{split}
\end{equation}
Here the estimate uses the positive quadratic form $\Tr(PY^2)$
and the fact that each $D_i$ commutes with $J$.

Third, the mixed term before its factor $1/u$ is $-w^TZ_cv$.
Taking covariance in $(I-E)w=Fv$ gives
$\Pi_{wv}=E\Pi_{wv}+F\Pi_{vv}$. Hence
\begin{equation}\label{eq:sm:mixed}
\begin{split}
 |\operatorname{Tr}(Z_c\Pi_{wv})|
 &\le |\operatorname{Tr}(Z_cF)|
      +\|E\|_{\mathrm{HS}}\|\Pi_{wv}\|_{\mathrm{HS}}
      +\|F\|_{\mathrm{HS}}\|I-\Pi_{vv}\|_{\mathrm{HS}}\\
 &\le 2h_*+\sqrt{h_*\beta_{\rm proj}}
        +2\sqrt{h_*(\beta_{\rm proj}+1)}\\
 &\le(2+\sqrt5+2\sqrt6)h_*.
\end{split}
\end{equation}
For its first term, diagonal similarity preserves diagonal entries, so
\[
 |\operatorname{Tr}(Z_cF)|
 =\left|\sum_i z_i\bigl((T_J)_{ii}-T_{ii}z_i\bigr)\right|
 \le2\sum_i T_{ii}\le2h_*.
\]
The other terms use \eqref{eq:sm:projectionblocks} and Cauchy--Schwarz.

Combining \eqref{eq:sm:negative}--\eqref{eq:sm:mixed} proves
\eqref{eq:sm:generalcurvature}. If its coefficient is positive, some
legal $y\ne0$ with $b^Ty=0$ has $y^TKy<0$, proving the assertion.
\end{proof}

\begin{corollary}[Curvature orthogonal to the current state]\label{thm:sm:curvature}
Use the smooth source \eqref{eq:sourceprofile} at a state where every live
coordinate has margin at least $a$. Suppose there is a nonzero live atom
and every nonzero live label satisfies $\ell_iq_i\le1$. Then there is a
real vector $g$, supported on those labels, such that
\begin{equation}\label{eq:sm:finalcurvature}
 \|g\|_2=1,\qquad x^Tg=0,\qquad
 g^T\nabla^2\mathscr F(x)g<0.
\end{equation}
Thus the raw Hessian restricted to the live part of $x^\perp$ has a
negative eigenvalue.
\end{corollary}
\begin{proof}
The source satisfies $c_i''\le-2u$. By the secant bounds,
\[
 |z_i|=|c_i'(x_i)|q_i\le\ell_iq_i\le1,
 \qquad
 r_i^2=c_iq_i^2\le\frac{c_i}{\ell_i^2}
 \le\frac{2u}{(9u/10)^2}=\frac{200}{81u}.
\]
Apply Theorem~\ref{thm:sm:general} with $L=200/81$ and
$b=D_d^{1/2}x$ on the nonzero live labels. Its coefficient is positive,
indeed greater than $1/2$, since
\[
 6(200/81)+7+\sqrt5+2\sqrt6<32=u/2.
\]
Some vector $y$ therefore satisfies $b^Ty=0$ and $y^TKy<0$.
Put $g_0=D_d^{1/2}y$. The live weights are positive, so $g_0\ne0$,
$x^Tg_0=0$, and
$g_0^T\nabla^2\mathscr F(x)g_0=2y^TKy<0$.
Normalize $g=g_0/\|g_0\|_2$ and extend it by zero to all other labels.
\end{proof}

The change of scale $D_d^{1/2}$ is used only to prove existence of the
direction. The algorithm normalizes the physical direction and uses
the raw Hessian restricted to $x^\perp$. This is what gives a fixed
increase $h^2$ in squared coefficient norm for either sign of the step.
The argument also rules out a zero-dimensional movement space on this
branch: the displayed negative direction belongs to that space. In
particular the case of a single nonzero coefficient cannot obstruct
movement after all outward tests have failed.

\section{Small movements and deterministic finite completion}
\label{sec:smallstep}

The curvature estimate can now be used without making progress depend
on the sign of a movement. A physical unit direction perpendicular to
$x$ increases $\|x\|_2^2$ by $h^2$ for either sign. We therefore choose
the sign by its spectral cost. The remaining deficit
$n-\|x\|_2^2$ pays for the permitted spectral drift, and bounds the
number of movements.

We first give the finite argument in terms of a common sufficiently
small $h>0$. The numerical sections provide its explicit polynomial
choice and implement every decision with the stated primitives.

\subsection{Boundary rounding and its total cost}
A label is frozen precisely when its coefficient is an endpoint.
Whenever a live label has $r_i^{\rm snap}=1-|x_i|\le\rho$, set
\begin{equation}\label{eq:rounding}
 s_i=\operatorname{sgn}_+(x_i),\qquad
 x_i\longleftarrow s_i.
\end{equation}
As $\rho<1$, this test never rounds a zero coordinate.
The source is always recomputed from the actual coordinates and hence
its $i$th weight becomes $c(s_i)=0$. The final signed atom stays in
the center $H(x)$.

\begin{lemma}[Near-boundary rounding]\label{lem:rounding}
If $x'$ results from rounding label $i$ in $x$, then
\begin{equation}\label{eq:roundingcost}
 \mathscr F(x')-\mathscr F(x)\le r_i^{\rm snap}\norm{A_i}.
\end{equation}
There are at most $n$ such updates. Define the remaining rounding
allowance, used only in the proof, by
\begin{equation}\label{eq:roundingallowance}
 \mathscr R(x)=\rho\sum_{i:\,|x_i|<1}\norm{A_i}.
\end{equation}
Then $\mathscr F+\mathscr R$ does not increase at a rounding, and
\begin{equation}\label{eq:allowancebound}
 0\le\mathscr R(x)\le\mathscr R(0)=\rho m
 \le\frac{\epsilon}{100n}\le\frac{\delta}{100n}.
\end{equation}
\end{lemma}
\begin{proof}
The doubled center changes by $r_i^{\rm snap}s_iJB_i$.
For every density $S$, its contribution to the objective increases
by at most $r_i^{\rm snap}\norm{JB_i}=r_i^{\rm snap}\norm{A_i}$.
The source loses a nonnegative summand. For each fixed full density,
fidelity is monotone in that source: sandwich by $S^{1/2}$ and use
operator monotonicity of the positive square root. See the integral proof in Lemma~\ref{lem:standardorder}.
The regularizer is unchanged at a fixed density. Comparing the two
objectives and then taking their maxima proves \eqref{eq:roundingcost}.

A frozen label never moves again. Rounding label $i$ decreases
$\mathscr R$ by $\rho\norm{A_i}$, which covers its potential increase
because $r_i^{\rm snap}\le\rho$. Finally rank one and Parseval give
$\sum_i\norm{A_i}=m\le n\epsilon$. Substituting
$\rho=1/(100n^2)$ and using $\epsilon\le\sqrt\epsilon=\delta$
proves \eqref{eq:allowancebound}.
\end{proof}

The allowance is a proof quantity. It records a cost that can occur
at most once per label and is constant on each open face. The
algorithm itself stores only the coefficients and numerical work
arrays, with center $J\widehat H(x)$ throughout.

After applying all available roundings, call the state \emph{prepared}.
Every live coordinate then has margin greater than $\rho$. An outward
movement of length $a=\rho/16$ remains interior. A physical unit live
vector $g$ satisfies $|g_i|\le1$, so both segments
$x+tg$, $|t|\le h\le a$, also remain on the current open face.
The live set, and therefore $\mathscr R$, is unchanged during either
movement; it changes only at the subsequent roundings.

\subsection{Choosing a sign by potential value}
Write
\[
 K_0(x)=\tfrac12\nabla^2\mathscr F(x)
\]
for the raw half-Hessian on the live coordinates. Unlike the balanced
matrix $K$ in \eqref{eq:sm:hessian}, this matrix uses the physical
coefficient scale.

\begin{lemma}[Symmetric potential estimate]\label{lem:symmetric}
Let $g$ be a real unit live vector and suppose
$|\frac{d^4}{dt^4}\mathscr F(x+tg)|\le M_{\rm loc}$ for $|t|\le h$.
Then
\begin{equation}\label{eq:symmetricdrift}
 \frac{\mathscr F(x+hg)+\mathscr F(x-hg)}2-\mathscr F(x)
 \le h^2g^TK_0(x)g+M_{\rm loc}h^4/24.
\end{equation}
If $q_+,q_-$ approximate the two candidate values to accuracy $\nu$,
choosing the candidate with smaller report gives an actual potential
increase at most the right-hand side of \eqref{eq:symmetricdrift}
plus $2\nu$.
\end{lemma}
\begin{proof}
For $F(t)=\mathscr F(x+tg)$, Taylor expansion through degree three
at $h$ and $-h$ gives
\[
 \frac{F(h)+F(-h)}2-F(0)=\frac{h^2}2F''(0)+R,
 \qquad |R|\le M_{\rm loc}h^4/24.
\]
Odd terms cancel, and $F''(0)=2g^TK_0(x)g$. If the chosen sign is $s$,
then for either sign $t$,
$F(sh)\le q_s+\nu\le q_t+\nu\le F(th)+2\nu$.
Thus $F(sh)\le\min\{F(h),F(-h)\}+2\nu$, which is bounded by their
average plus $2\nu$.
\end{proof}

This comparison uses symmetry to control spectral cost while retaining
a deterministic choice of sign. The angle with the potential gradient
does not need to be estimated: its linear contribution cancels in the
average, and the chosen report is no larger than either alternative.

\subsection{Progress for every movement}
Define the squared-norm deficit
\begin{equation}\label{eq:radialdeficit}
 U_{\rm rad}(x)=n-\|x\|_2^2=\sum_i(1-x_i^2).
\end{equation}
It is nonnegative on the cube, is initially $n$, and is zero precisely
at its vertices.

\begin{lemma}[Deterministic progress]\label{lem:radialprogress}
A length-$a$ outward movement decreases $U_{\rm rad}$ by at least $a^2$.
If $\|g\|_2=1$ and $x^Tg=0$, then either movement $x\mapsto x\pm hg$
decreases $U_{\rm rad}$ by exactly $h^2$. Boundary rounding cannot
increase $U_{\rm rad}$.
\end{lemma}
\begin{proof}
An outward change of coordinate $i$ increases squared norm by
$2a|x_i|+a^2\ge a^2$. Orthogonality gives
\[
 \|x\pm hg\|_2^2=\|x\|_2^2\pm2h x^Tg+h^2\|g\|_2^2
 =\|x\|_2^2+h^2.
\]
Rounding increases one squared coordinate from $x_i^2$ to one.
\end{proof}

\subsection{Finite completion and spectral control}
\begin{definition}[Analytic walk properties]\label{def:analytictrial}
Put $\beta=\delta/(100n)$ and choose $0<h\le a$. A walk starts at
$x=0$, stops when all labels are frozen or after
$T=\lceil16n/h^2\rceil$ movements, and has the following properties.
\begin{enumerate}
\item It performs all available roundings \eqref{eq:rounding} before
its first movement and after every subsequent movement.
\item At every nonterminal prepared state it makes either one
length-$a$ outward increment, or one movement $x\mapsto x\pm hg$
with a real unit live vector $g$ satisfying $x^Tg=0$.
\item The chosen movement, measured before rounding, has actual
potential increase at most $\beta h^2$.
\end{enumerate}
Proposition~\ref{prop:exacttrial} realizes these properties with exact
choices. Proposition~\ref{prop:implementedtrial} verifies them for
Algorithm~\ref{alg:trial}.
\end{definition}

\begin{lemma}[Finite completion]\label{lem:finitefinish}
Every walk satisfying Definition~\ref{def:analytictrial} completes
before its horizon.
\end{lemma}
\begin{proof}
Each nonterminal movement increases squared norm by at least $h^2$,
since $a\ge h$, and the subsequent roundings only increase it further.
If the walk were still unfinished after $T$ movements, its squared
norm would be at least $Th^2\ge16n>n$, contradicting $x\in[-1,1]^n$.
\end{proof}

\begin{lemma}[A nonincreasing spectral account]\label{lem:ledger}
For such a walk, define
\begin{equation}\label{eq:ledger}
 W(x)=\mathscr F(x)+\mathscr R(x)+\beta U_{\rm rad}(x).
\end{equation}
Every movement and every rounding leave $W$ nonincreasing. At every
state, including the final signing,
\begin{equation}\label{eq:ledgerbound}
 \|H(x)\|\le\mathscr F(x)\le W(x)\le W(0)<35\delta.
\end{equation}
\end{lemma}
\begin{proof}
During a movement the rounding allowance is constant. The potential
increase is at most $\beta h^2$, while
Lemma~\ref{lem:radialprogress} decreases $\beta U_{\rm rad}$ by at
least that amount. At a rounding, Lemma~\ref{lem:rounding} controls
$\mathscr F+\mathscr R$ and $U_{\rm rad}$ decreases. Finally the
potential dominates discrepancy, the other two terms are nonnegative,
and
\[
 W(0)\le34\delta+\frac{\delta}{100n}+\frac{\delta}{100}
 \le34.02\delta<35\delta.
\]
\end{proof}

\begin{theorem}[Deterministic signing]\label{thm:trial}
A walk satisfying Definition~\ref{def:analytictrial} returns a full
signing with discrepancy less than $35\delta$.
\end{theorem}
\begin{proof}
Combine finite completion with \eqref{eq:ledgerbound}.
\end{proof}

\subsection{Existence without computational assumptions}
\begin{proposition}[Exact local choices]\label{prop:exacttrial}
For each fixed input with $m>0$, some common $h>0$ and exact choices
satisfy Definition~\ref{def:analytictrial}. Thus the signing and
partition existence conclusions do not assume Model~\ref{model:solver}.
\end{proposition}
\begin{proof}
At a prepared state, use an outward candidate whose exact potential
is no larger than the current value, whenever such a candidate exists.
Otherwise every outward candidate strictly increases the potential.
A zero atom would leave it unchanged, so every live atom is nonzero.
Lemma~\ref{lem:sm:outward} forces the failed secant inequalities,
and Corollary~\ref{thm:sm:curvature} supplies a physical unit direction
$g$ with $x^Tg=0$ and $g^TK_0(x)g<0$. Choose the sign with the smaller
exact candidate value.

For fixed input there is a finite uniform fourth-derivative bound on
all relevant segments. Fix a live set and endpoint signs for its
complement, restrict live coordinates to margin at least $\rho/2$,
and take physical directions with $\|g\|_2=1$. These are compact sets.
Positive live weights keep the source support fixed; by
Lemma~\ref{lem:transport}, the optimized potential is smooth in a
neighborhood of every such state. Its fourth directional derivatives
have a finite maximum on this compact set. There are only finitely
many live sets and endpoint choices, so one bound $M_0<\infty$ suffices.

Choose $0<h\le a$ with $M_0h^2/24\le\beta$. In the curvature branch,
Lemma~\ref{lem:symmetric} with exact reports gives a chosen potential
increase at most $\beta h^2$. The outward branch has nonpositive
increase. Both branches make deterministic progress, and the roundings
have the prescribed cost. Selecting a direction at each encountered
state yields the required finite walk.
\end{proof}

The implementation needs no lower bound on the negative eigenvalue.
A small positive curvature tolerance still leaves enough of the
$\beta h^2$ account to pay for value errors and the fourth-order
remainder. The next sections replace the compactness bound by explicit
polynomial scales and compute each local choice.

\section{The exact semidefinite value problem}\label{sec:sdp}

The analytic argument used the potential and its derivatives. The
implementation obtains the required information from nearby values.
Two block-matrix identities convert the density maximum into a finite
SDP. We prove them and the corresponding dual equality, so the
optimization queried by the algorithm is the same one whose response
was analyzed above. The Schur-complement reformulations follow the
standard SDP framework of Boyd--Vandenberghe
\cite[Section~4.6.2 and Appendix~A.5.5]{boydvandenberghe2004}.

\begin{assumption}[Value and spectral primitives]\label{model:solver}
Real arithmetic, comparisons, and nonnegative scalar square roots have
unit cost. Complex entries use two real registers. Fix constants $A_*,b_*\in\mathbb N$ and a
deterministic solver that, for every feasible finite-valued SDP
constructed here, of scalar data size $L$, returns an additive-$\nu$
value in at most $A_*(L+1+\nu^{-1})^{b_*}$ operations.
A second deterministic primitive returns an exact orthogonal
eigendecomposition of a real symmetric matrix, with a fixed choice of
basis and order at ties. Its invocation has unit cost; forming its
input, reading its output, and all subsequent scalar work are counted.
\end{assumption}
The runtime bounds count operations in this real-arithmetic model,
with the polynomial depending only on the fixed primitive bounds.
No binary-encoding bound is asserted.

\begin{lemma}[Two semidefinite identities]\label{lem:sdpidentities}
For $D\times D$ positive semidefinite matrices $S,T$, the following
identities hold, including at singular inputs. In the first formula
$W\in\C^{D\times D}$ is unrestricted; in the second, $Y$ is
Hermitian positive semidefinite.
\begin{equation}\label{eq:O5}
 \boxed{\quad
 \operatorname{Fid}(S,T)
 =\max_W\left\{\operatorname{Re}\operatorname{Tr}W:
     \begin{pmatrix}S&W\\W^*&T\end{pmatrix}\succeq0\right\}.
 \quad}
\end{equation}
\begin{equation}\label{eq:O6}
 \boxed{\quad
 \operatorname{Tr}\sqrt S
 =\max_{Y=Y^*\succeq0}\left\{\operatorname{Tr}Y:
       \begin{pmatrix}S&Y\\Y&I_D\end{pmatrix}\succeq0\right\}.
 \quad}
\end{equation}
Both maxima are attained. The first is the fidelity SDP of
\cite[Theorem~3.17 and equation~(3.110)]{watrous2018}.
\end{lemma}
\begin{proof}
To prove it first take \(S,T\succ0\). Congruence by
\(\operatorname{diag}(S^{-1/2},T^{-1/2})\) shows that block
positivity is equivalent to \(W=S^{1/2}CT^{1/2}\) with
\(\|C\|\le1\): the Schur complement is \(I-C^*C\succeq0\).
The maximum of
\(\operatorname{Re}\operatorname{Tr}(T^{1/2}S^{1/2}C)\)
over contractions is its trace norm. Indeed a singular-value
decomposition bounds this quantity by the sum of singular values,
and the adjoint polar factor attains that sum. This trace norm is
\(\operatorname{Tr}\sqrt{S^{1/2}TS^{1/2}}\).

If \(S\) or \(T\) is singular, positivity of the block forces
\(\ker S\subseteq\ker W^*\) and \(\ker T\subseteq\ker W\).
For example, apply its nonnegative quadratic form to \((v,t w)\)
with \(v\in\ker S\) and both signs and phases of small \(t\).
Restrict to the two supports and apply the preceding proof there,
extending the resulting contraction by zero. Equivalently one can
use inverses on these supports. Thus \eqref{eq:O5} is exact at singular inputs.

The Schur complement gives \(Y^2\preceq S\). Operator monotonicity
of the square root gives \(Y\preceq\sqrt S\), hence the upper bound.
It is attained by $Y=\sqrt S$. The Schur-complement and order
facts used here were proved in Lemma~\ref{lem:standardorder}.
\end{proof}
\begin{proposition}[The actual value is an SDP]\label{prop:sdp}
Set $M_x=J\widehat H(x)$. The potential equals the
following optimum over $D\times D$ matrices:
\begin{equation}\label{eq:O7}
 \begin{array}{ll}
 \text{maximize}_{S=S^*,\,Y=Y^*,\,W}
 &\operatorname{Tr}(M_xS)+2\operatorname{Re}\operatorname{Tr}W
                         +2\theta\operatorname{Tr}Y\\[2mm]
 \text{subject to}
 &\operatorname{Tr}S=1,\qquad Y\succeq0,\\[1mm]
 &\begin{pmatrix}S&W\\W^*&\Omega_x(S)\end{pmatrix}\succeq0,
 \qquad
 \begin{pmatrix}S&Y\\Y&I_D\end{pmatrix}\succeq0.
 \end{array}
\end{equation}
Its direct real formulation has $4D^2-1=16m^2-1$ affine variables
and one block-diagonal pencil of order $10D=20m$.
\end{proposition}
\begin{proof}
Equations~\eqref{eq:O5}--\eqref{eq:O6} identify the two auxiliary
maximizations at every density.
The top-left blocks already impose \(S\succeq0\). The objective
and every constraint are affine in the real coordinates of the
Hermitian matrices \(S,Y\) and the complex matrix \(W\). For each
fixed \(S\), the two maximizations over \(W,Y\) are independent
and attain \eqref{eq:O5}--\eqref{eq:O6}, proving equality with the actual potential \eqref{eq:potential}.
There are \(4D^2-1\) real affine variables, two PSD blocks of size
\(2D\), and one PSD block of size \(D\). Complex Hermitian cones
may be replaced by real symmetric cones via the usual realification.
For \(\lambda>0\), replace only the lower-right source block in
\eqref{eq:O7} by \(\Omega_x(S)+\lambda I_D\), and call its optimum
\(\mathscr F_{\lambda}(x)\). Then
\begin{equation}\label{eq:O8}
 0\le\mathscr F_{\lambda}(x)-\mathscr F(x)
       \le2\sqrt{\lambda D}.
\end{equation}
The perturbation bound is uniform in $S$, including singular densities.
Set $A=S^{1/2}\Omega_x(S)S^{1/2}$ and write
$\|C\|_1=\Tr\sqrt{C^*C}$ for the trace norm of a possibly rectangular
matrix $C$. The trace norm of a horizontal
block matrix gives
\[
 \begin{split}
 \operatorname{Tr}\sqrt{A+\lambda S}
 &=\left\|[A^{1/2},\sqrt\lambda S^{1/2}]\right\|_1\\
 &\le\|[A^{1/2},0]\|_1+
            \|[0,\sqrt\lambda S^{1/2}]\|_1\\
 &=\operatorname{Tr}\sqrt A+
            \sqrt\lambda\operatorname{Tr}\sqrt S\\
 &\le\operatorname{Tr}\sqrt A+\sqrt{\lambda D}.
 \end{split}
\]
The trace-norm triangle inequality follows immediately from the same
contraction duality used in \eqref{eq:O5}. Monotonicity supplies the lower
bound. Multiply by two and maximize over the unchanged trace-one
set to obtain \eqref{eq:O8}.

For requested error \(0<\nu\le1\), choose
\begin{equation}\label{eq:O9}
 \lambda=\frac{\nu^2}{16D}.
\end{equation}
The perturbation error is at most \(\nu/2\). Therefore a value of
the regularized SDP accurate to \(\nu/2\) is accurate to \(\nu\)
for the original actual potential. This regularization is used only
inside the value oracle; it does not replace the source or its weights in
the mathematical walk.
\end{proof}
\begin{remark}[Which value program the algorithm queries]\label{rem:valueprogram}
The algorithm in Theorem~\ref{thm:main} uses the exact program
\eqref{eq:O7} under Model~\ref{model:solver}. The positive source
perturbation in \eqref{eq:O8}--\eqref{eq:O9} is an optional realization
of accurate values by a strictly feasible program. It is internal to
that realization and is never added to the source weights in the walk.
\end{remark}
The programs have polynomial-size scalar descriptions. Their explicit
coordinate construction and a strictly feasible perturbed realization
are given in Appendix~\ref{sec:sdpdetails}. The next duality argument
uses the value perturbation bound \eqref{eq:O8} to handle source
support loss.

\subsection{A dual derived from the same variational identities}
The value representation also has an explicit dual. We derive equality
of the two values directly, including singular sources, to clarify
which optimization is being differentiated in the analytic proof.
The Lagrangian calculation follows the standard framework for
generalized inequalities and SDP duality
\cite[Section~5.9]{boydvandenberghe2004}; equality for our potential
is established below by an explicit transport construction.
In this subsection write $M_x=J\widehat H(x)$ and keep
$x$ fixed.

\begin{proposition}[Transport and semidefinite dual]\label{prop:dual}
For full positive definite $D\times D$ matrices $Z,Q$, the potential
equals the following infimum:
\begin{equation}\label{eq:transportdual}
 \mathscr F(x)=\inf_{Z\succ0,\ Q\succ0}
 \left\{\lambda_{\max}(M_x+Z^{-1}+\Omega_x(Z)+Q)
                         +\theta^2\Tr Q^{-1}\right\}.
\end{equation}
The inverse variables in this infimum occupy the full $D$-dimensional
space; an infimum need not be attained when the source is singular.
Equivalently, it is the semidefinite dual value
\begin{equation}\label{eq:conicdual}
 \begin{array}{ll}
 \text{infimize}_{\lambda\in\R,\ P=P^*,\ Z=Z^*,\ Q=Q^*,\ R=R^*}
       &\lambda+\Tr R\\[1mm]
 \text{subject to}
       &M_x+P+\Omega_x(Z)+Q\preceq\lambda I_D,\\[1mm]
       &\begin{pmatrix}P&-I_D\\-I_D&Z\end{pmatrix}\succeq0,
        \qquad
        \begin{pmatrix}Q&-\theta I_D\\-\theta I_D&R\end{pmatrix}
                                                       \succeq0.
 \end{array}
\end{equation}

\end{proposition}
\begin{proof}
The transport inequality, valid for every full $Z\succ0$, gives
\[
 2\Fid(S,\Omega_x(S))
 \le\Tr(SZ^{-1})+\Tr(\Omega_x(S)Z)
 =\Tr\bigl(S(Z^{-1}+\Omega_x(Z))\bigr).
\]
The last equality is self-adjointness of the explicit source map.
The same variational identity applied to $S$ and $\theta^2I_D$ gives,
for every $Q\succ0$,
\[
 2\theta\Tr\sqrt S\le\Tr(SQ)+\theta^2\Tr Q^{-1}.
\]
Maximizing their sum over trace-one densities gives the right side of
\eqref{eq:transportdual} as an upper bound on $\mathscr F(x)$.
This is weak duality, proved without interchanging a maximum and a
minimum.

Suppose first that the source is positive definite at every faithful
density, as happens when its support is full. Let $S_*$ be the faithful
maximizer and $Z_*$ its full positive transport. Choose
$Q_*=\theta S_*^{-1/2}$. Density stationarity \eqref{eq:9} reads
\[
 M_x+Z_*^{-1}+\Omega_x(Z_*)+Q_*=\lambda_* I_D.
\]
Moreover
$\lambda_*=\mathscr F(x)-\theta\Tr\sqrt{S_*}$ and
$\theta^2\Tr Q_*^{-1}=\theta\Tr\sqrt{S_*}$.
The upper bound is therefore attained with value exactly
$\mathscr F(x)$.

For a possibly singular source, introduce the linear self-adjoint map
\[
 \Omega_\eta(Y)=\Omega_x(Y)+\eta\Tr(Y)I_D,\qquad\eta>0.
\]
On trace-one densities this is the source perturbation in
\eqref{eq:O8}. Its value $\mathscr F_{\eta}(x)$ satisfies
$0\le\mathscr F_{\eta}(x)-\mathscr F(x)\le2\sqrt{\eta D}$.
Let $S_\eta$ be its optimizing density, $Z_\eta$ its full transport,
and $\lambda_\eta$ the multiplier for $\Tr S_\eta=1$ in density
stationarity. Set $Q_\eta=\theta S_\eta^{-1/2}$. The preceding stationarity argument
gives
\[
 M_x+Z_\eta^{-1}+\Omega_x(Z_\eta)+Q_\eta
   =(\lambda_\eta-\eta\Tr Z_\eta)I_D.
\]
Consequently the objective in \eqref{eq:transportdual} at these
variables is
$\mathscr F_{\eta}(x)-\eta\Tr Z_\eta\le\mathscr F_{\eta}(x)$.
Weak duality puts it above $\mathscr F(x)$. Letting $\eta\downarrow0$
proves equality of the infimum with $\mathscr F(x)$ in all cases.
This also explains why a full transport may diverge outside the source
support while the supported analytic transport remains finite.

We next derive the displayed conic dual from a Lagrangian, rather
than merely recognizing its Schur complements. Keep the primal fidelity auxiliary matrix $W$. Replace the primal
Hermitian positive variable $Y$ by an unrestricted complex matrix $V$,
use the block $\left(\begin{smallmatrix}S&V\\V^*&I_D\end{smallmatrix}\right)$,
and replace $\Tr Y$ in the objective by $\Re\Tr V$.
This is an equivalent value program: \eqref{eq:O5} with $T=I_D$
shows that its maximum over $V$ is $\Tr\sqrt S$, attained at
$V=\sqrt S$, which was already feasible in the original program.
Keep the redundant constraint $S\succeq0$ explicitly. Thus the primal
value is unchanged, though the auxiliary coordinates used only for
this dual derivation differ from the implemented program.

Introduce positive semidefinite block multipliers
\[
 U_1=\begin{pmatrix}P&C\\C^*&Z\end{pmatrix},\qquad
 U_2=\begin{pmatrix}Q&E\\E^*&R\end{pmatrix}
\]
and a real scalar $\lambda$ for $\Tr S=1$. The upper-bounding
Lagrangian is the primal objective plus the trace pairings with the
two positive blocks and $\lambda(1-\Tr S)$. Expanding it gives
\[
\begin{split}
 \mathcal L_{\rm dual}(S,W,V)
 ={}&\lambda+\Tr R
 +\Tr\bigl((M_x+P+\Omega_x(Z)+Q-\lambda I_D)S\bigr)\\
 &+2\Re\Tr((I_D+C^*)W)
   +2\Re\Tr((\theta I_D+E^*)V).
\end{split}
\]
Here the symbol $\mathcal L_{\rm dual}$ denotes only this Lagrangian.
Its supremum over unrestricted complex $W,V$ is finite precisely when
$C=-I_D$ and $E=-\theta I_D$. Its supremum over $S\succeq0$ is
finite precisely when the coefficient of $S$ is negative semidefinite;
otherwise a rank-one $S$ of arbitrarily large size makes it diverge.
Under these conditions the supremum equals $\lambda+\Tr R$.
Minimizing over the multipliers gives exactly \eqref{eq:conicdual}.
In particular the sign, transpose, and factor two in each off-diagonal
block follow from the trace pairing. This calculation proves weak
conic duality directly.

Finally, block positivity in the first block constraint of
\eqref{eq:conicdual} forces $Z\succ0$, because its off-diagonal block
has no kernel. The Schur complement then says $P\succeq Z^{-1}$.
Similarly the second block forces $Q\succ0$ and
$R\succeq\theta^2Q^{-1}$, since $\theta>0$.
Thus every feasible dual tuple has objective at least the expression
in \eqref{eq:transportdual} for its $Z,Q$. Conversely, for any such
$Z,Q$, take
\[
 P=Z^{-1},\qquad R=\theta^2Q^{-1},\qquad
 \lambda=\lambda_{\max}(M_x+P+\Omega_x(Z)+Q).
\]
Both blocks are positive semidefinite by their Schur complements and
the first constraint holds. This proves equality of the two infima. Combined with the preceding
explicit transport construction, it also proves that this Lagrangian
dual has no value gap with the primal. Full source support gives an
attaining dual tuple; for singular sources the positive perturbations
give a sequence of feasible dual tuples whose values converge to the
primal optimum. The approximating dual tuples establish equality even when the
unperturbed primal has no strictly feasible point.
\end{proof}

The dual variables recover the density--transport stationarity system
used in the analytic argument. The primal, meanwhile, supplies the
value interface used by the algorithm. This separates the role of
the variational equations in the proof from the information that must
be returned by the solver.

\section{The finite value-query walk}\label{sec:walk}

We verify Algorithm~\ref{alg:trial}. Its value tolerances ensure that
an accepted outward report has small true cost, while rejection of
every report preserves the curvature hypotheses. Finite differences
then approximate the raw Hessian. Restricting it to the space
orthogonal to the current coefficients makes progress independent of
the selected sign.

The scales are those in \eqref{eq:algorithmsetup}, with $M$ given by
\eqref{eq:regularitybudget}. Appendix~\ref{sec:regularity} proves that,
at a prepared state and along a live direction $b$ with $|b_i|\le2$,
the second through fourth derivatives of $\mathscr F(x+tb)$ are bounded
by $M$ for $|t|\le\rho/16$. This covers all value stencils and movement
segments used below.

\subsection{Local tests imply constrained curvature}
At a prepared state put $s_i=\operatorname{sgn}_+(x_i)$ and consider
$x^{(i)}=x+as_ie_i$ for each live label. Query the actual potential at
$x$ and these candidates to accuracy $\nu_{\rm loc}=\beta h^2/8$.

\begin{lemma}[The local value test]\label{lem:localtest}
An accepted outward movement has actual potential increase at most
$\beta h^2/2$. If every candidate is rejected, then every live atom is
nonzero and $\ell_iq_i<1$ for every live label, where
$\ell_i=[c(x_i)-c(x_i+as_i)]/a$. In that case the raw Hessian restricted
to the live part of $x^\perp$ has a negative eigenvalue.
\end{lemma}
\begin{proof}
Two value errors contribute at most $2\nu_{\rm loc}$ to a difference.
Acceptance at reported difference at most $2\nu_{\rm loc}$ therefore
implies an actual increase at most $4\nu_{\rm loc}=\beta h^2/2$.
Rejection implies a strictly positive actual difference.
Lemma~\ref{lem:sm:outward} says that $\ell_iq_i\ge1$ would make that
candidate nonincreasing, so this inequality must fail. A zero atom
would leave the potential unchanged, and its reported difference
would be at most $2\nu_{\rm loc}$; thus no zero atom remains in this
branch. Corollary~\ref{thm:sm:curvature} now applies.
\end{proof}

\subsection{Computing the tangent frame}
Let $I=\{i_1<\cdots<i_k\}$ be the live labels and let
$z=(x_{i_1},\ldots,x_{i_k})\in\R^k$. The matrix $Q_x$ defined in
Section~\ref{sec:earlyalgorithm} has either $k$ columns, when $z=0$,
or $k-1$ columns otherwise.

\begin{lemma}[A scalar construction of the frame]\label{lem:tangentframe}
The columns of $Q_x$ are an orthonormal basis of $z^\perp$.
It can be computed using polynomially many scalar arithmetic operations,
comparisons, and nonnegative square roots. No spectral decomposition
is needed for this construction.
\end{lemma}
\begin{proof}
The claim is immediate for $z=0$. Otherwise let
$q=z/\|z\|_2$, $s=\operatorname{sgn}_+(q_1)$ and $w=q+se_1$.
Then
\[
 \|w\|_2^2=2(1+s q_1)=2(1+|q_1|)\ge2.
\]
The matrix $R=I-2ww^T/\|w\|_2^2$ is symmetric and satisfies $R^2=I$.
Direct substitution gives $Re_1=-sq$. Its other columns are therefore
orthogonal to $q$ and form an orthonormal basis of $q^\perp=z^\perp$.
These are precisely the columns retained in $Q_x$.

The construction computes a Euclidean norm, one sign, the entries of
$w$, and the displayed rank-one formula for $R$. The normalization
uses $\|z\|_2>0$ only on the nonzero branch, and the remaining denominator
is at least two. All operations belong to Model~\ref{model:solver}.
\end{proof}

If all outward tests fail, Lemma~\ref{lem:localtest} guarantees a nonzero
direction in $z^\perp$. In particular the compressed matrix below has
positive order, so selecting its smallest eigenvalue is well-defined.

\subsection{Finite differences and direction selection}\label{sec:stencils}
\begin{lemma}[A finite approximate curvature direction]\label{lem:direction}
At a prepared state where every local value test fails, the value
stencils \eqref{eq:algorithmstencils}, tangent frame, and one exact
eigendecomposition produce a real unit live vector $g$ satisfying
\[
 x^Tg=0,\qquad g^TK_0(x)g<\beta/2,
 \qquad K_0(x)=\tfrac12\nabla^2\mathscr F(x).
\]
\end{lemma}
\begin{proof}
A centered diagonal stencil has error at most
$Mt^2/12+4\nu_{\rm H}/t^2$. The four-point mixed stencil has error at
most $Mt^2/12+\nu_{\rm H}/t^2$, by symmetric Taylor expansion along
$e_i+e_j$ and $e_i-e_j$. All stencil points remain on the live face.
The real symmetric Hessian estimate $\widetilde{\mathcal H}$ therefore
satisfies, with $\widetilde K=\widetilde{\mathcal H}/2$,
\[
 \|\widetilde K-K_0(x)\|
 \le\frac{k}{2}\left(\frac{Mt^2}{12}
                         +\frac{4\nu_{\rm H}}{t^2}\right)
 \le\frac{\beta}{16}.
\]
Here we used $k\le n$, $Mt^2\le\beta/(128n)$, and
$\nu_{\rm H}=\beta t^2/(128n)$.

Put $B=Q_x^TK_0(x)Q_x$ and $\widetilde B=Q_x^T\widetilde KQ_x$.
Since $Q_x$ is an isometry, $\|B-\widetilde B\|\le\beta/16$.
Lemma~\ref{lem:localtest} gives $\lambda_{\min}(B)<0$.
Apply the exact EVD primitive to $\widetilde B$ and take its unit
column $v$ with the smallest signed eigenvalue, using the stipulated
tie rule. Extend $g=Q_xv$ by zero on frozen labels. Then
\[
 \|g\|_2=1,\quad x^Tg=0,\qquad
 g^TK_0(x)g=v^TBv
 \le\lambda_{\min}(B)+2\beta/16<\beta/2.
\]
The two error terms compare the selected Rayleigh value and the
minimum eigenvalue of the true and estimated compressed matrices.
\end{proof}

\begin{remark}[Hessian evaluation by implicit differentiation]
\label{rem:implicit-hessian}
The value stencils are one way to compute the curvature.
At a fixed open live face, given the optimizing density $S$ and
supported transport $Z$, one can instead differentiate the nonlinear
stationarity equations \eqref{eq:nonlinearresponse} once and solve
\eqref{eq:coupledresponse} for their first responses $X,U,\dot\lambda$.
Substitution in \eqref{eq:second-envelope}, including its explicit
source-curvature term, gives $\mathscr F''(x)[h,h]$; polarization then
recovers $K_0=\tfrac12\nabla^2\mathscr F(x)$.
The solve keeps $X$ on the full trace-zero density tangent and
restricts only the transport response $U$ to the source support.
No second derivatives of the optimizers are needed.

To replace the value-stencil routine in the algorithm, a numerical
implementation of these solves would need to certify
$\|\widetilde K-K_0\|\le\beta/16$, accounting for optimizer and
linear-solve errors, with polynomial work in Model~\ref{model:solver}.
The finite-difference implementation above establishes this accuracy
using only SDP value queries.
\end{remark}

\begin{lemma}[Implemented deterministic drift]\label{lem:movement}
Every implemented movement, measured before subsequent rounding, has
actual potential increase at most $\beta h^2$.
\end{lemma}
\begin{proof}
For a curvature movement, Lemmas~\ref{lem:symmetric} and
\ref{lem:direction}, together with the two candidate reports of
accuracy $\nu_{\rm loc}$, give
\[
 \mathscr F(x_{\rm new})-\mathscr F(x)
 \le\frac{\beta h^2}{2}+\frac{Mh^4}{24}+2\nu_{\rm loc}
 \le\left(\frac12+\frac1{24}+\frac14\right)\beta h^2
 <\beta h^2.
\]
We used $Mh^2\le\beta$, which follows from the chosen movement scale.
The outward branch has increase at most $\beta h^2/2$ by
Lemma~\ref{lem:localtest}. The subsequent rounding cost is covered
separately by $\mathscr R$ in Lemma~\ref{lem:rounding}.
\end{proof}

\subsection{Verification of the full walk}
\begin{proposition}[Implementation of the analytic walk]
\label{prop:implementedtrial}
Algorithm~\ref{alg:trial} satisfies Definition~\ref{def:analytictrial}.
\end{proposition}
\begin{proof}
Preparation performs finitely many near-boundary roundings with the
cost in Lemma~\ref{lem:rounding}. A movement is either the specified
outward step or a physical unit step orthogonal to $x$, by
Lemma~\ref{lem:direction}. Since $h\le a=\rho/16$ and $|g_i|\le1$,
the entire segments remain in the current open face and keep frozen
coordinates fixed. Lemma~\ref{lem:movement} bounds the actual spectral
drift. The algorithm uses the common step size and horizon in
Definition~\ref{def:analytictrial}, so all its properties hold.
\end{proof}

\begin{proof}[Proof of Theorem~\ref{thm:main}]
For $m=0$ return the all-positive signing. Otherwise
$0<\epsilon\le\varepsilon$. Proposition~\ref{prop:implementedtrial}
and Theorem~\ref{thm:trial} show that the algorithm completes and
returns discrepancy less than $35\delta\le35\sqrt\varepsilon$.
Theorem~\ref{thm:runtime} supplies the arithmetic bound.
Outward displacements have size $a=\rho/16$, and physical curvature
steps have coordinate magnitudes at most $h\le\rho/16$.
Each final rounding occurs within distance $\rho$ of the selected
endpoint, after which that coordinate stays fixed.
\end{proof}

\section{AI Usage}
The author conceived the approach of using operator-valued
$K$-transforms, or variants based on regularizers such as Tsallis--$1/2$
rather than log-determinants, as discrepancy potentials in 2022--2023.
The semidefinite representation of the Tsallis--$1/2$ regularized
potential was subsequently conceived by an internal version of Google's
Gemini, used with a specific harness. Extensive calculations with
ChatGPT 5.5--5.6 Sol Ultra helped develop and test proof strategies,
and rule out simpler approaches, leading to the strategy presented here.
GPT 6 Astra Ultra was used to develop Lean formalizations of the
Matrix Spencer and Weaver/Kadison--Singer arguments, including separate
existence and algorithmic formulations.

\section{Acknowledgements}
The idea of using free probability to approach Weaver's discrepancy
problem, potentially algorithmically, was encouraged by inspiring
conversations with Adam W. Marcus in 2016--2017. The author is
profoundly grateful for his encouragement. Numerous conversations and
continued encouragement from Nikhil Srivastava have been invaluable
over the years, as were early discussions with Nick Ryder and Jonathan
Leake.

The author thanks Daniel A. Spielman for hosting him at Yale University
in the summer of 2019, and for extensive discussions of scalar-valued
free probability and its possible applications to matrix discrepancy.
The approach began to crystallize in the summer of 2023, during a
visit hosted by Ramon van Handel at Princeton University. Ramon's
persistent questions about concrete special cases, bottlenecks, and
obstacles pushed the author to examine the approach carefully and
eliminate many of those obstacles, strengthening his conviction that
it could succeed. The author warmly thanks Ramon for his hospitality,
many discussions of operator-valued free probability, and insistence
on an exceptionally high but necessary standard of clarity and rigor.

Finally, the author thanks Tibo from OpenAI for providing numerous
Codex resets to paid Codex subscribers, which were invaluable in
helping this project go \emph{Ultra Fast}.

\appendix
\section{Standard matrix and analytic facts}\label{sec:preliminaries}

This section states the standard results used in the proof. The matrix
order and fidelity facts are finite-dimensional; their use does not
require a background in quantum mechanics. The only complex-analytic
inputs are the principal matrix square root and Cauchy's estimate.
The source-specific stationarity and duality arguments are derived in
the subsequent sections.

\subsection{Positive matrices and fidelity}
\begin{lemma}[Schur complement and square-root order]\label{lem:standardorder}
Let $A=A^*$ and $C\succ0$ be square matrices, and let $X$ have the
compatible rectangular size. Then
\[
 \begin{pmatrix}A&X\\X^*&C\end{pmatrix}\succeq0
 \quad\Longleftrightarrow\quad A-XC^{-1}X^*\succeq0.
\]
For same-size positive semidefinite matrices $A,B$,
$A\preceq B$ implies $A^{1/2}\preceq B^{1/2}$.
\end{lemma}
\begin{proof}
For the first assertion use the congruence factorization
\[
 \begin{pmatrix}A&X\\X^*&C\end{pmatrix}
 =\begin{pmatrix}I&XC^{-1}\\0&I\end{pmatrix}
   \begin{pmatrix}A-XC^{-1}X^*&0\\0&C\end{pmatrix}
   \begin{pmatrix}I&0\\C^{-1}X^*&I\end{pmatrix}.
\]
This is the Schur-complement criterion; see
\cite[Theorem~1.3.3]{bhatia2007}. The displayed factorization also
covers an indefinite $A$ whenever $C\succ0$.
For the second assertion, the spectral theorem gives
\[
 A^{1/2}=\frac1\pi\int_0^\infty t^{-1/2}A(A+tI)^{-1}\,dt.
\]
For positive definite matrices, inversion reverses order, as follows
by congruence and scalar inversion of eigenvalues. Consequently
$A(A+tI)^{-1}=I-t(A+tI)^{-1}$ preserves order for $t>0$.
Integrating proves the assertion, with singular endpoints obtained by
continuity. This is the exponent-$1/2$ case of
\cite[Theorem~1.5.9]{bhatia2007}.
\end{proof}

\begin{theorem}[Standard fidelity formulas]\label{thm:standardfidelity}
For same-size positive semidefinite matrices $P,Q$, define
$\Fid(P,Q)=\Tr\sqrt{Q^{1/2}PQ^{1/2}}$. Then
\[
 \Fid(P,Q)=\max_W\left\{\Re\Tr W:
          \begin{pmatrix}P&W\\W^*&Q\end{pmatrix}\succeq0\right\},
\]
where $W$ is an unrestricted complex matrix of that size, and
\[
 2\Fid(P,Q)=\inf_{Z\succ0}\{\Tr(PZ^{-1})+\Tr(QZ)\}.
\]
The first maximum is attained. When $P,Q\succ0$, the second infimum
has a unique minimizer, characterized by $ZQZ=P$.
For singular inputs the infimum need not be attained.
\end{theorem}
These are \cite[Theorems~3.17 and~3.19, equation~(3.110)]{watrous2018}.
We give direct proofs in Lemmas~\ref{lem:transport} and
\ref{lem:sdpidentities}, including the support conventions used here.
The word fidelity always denotes the unsquared quantity in this paper.

For each positive integer $r$, let $\operatorname{id}_r$ be the identity
map on $r\times r$ matrices. A linear map $\Psi$ on complex matrices
is \emph{completely positive} if $\operatorname{id}_r\otimes\Psi$
preserves positive semidefinite matrices for every $r$.
For a finite family of complex matrices $C_a$, a representation
$\Psi(Y)=\sum_a C_aYC_a^*$ implies this property, since each amplified
summand is a congruence by $I_r\otimes C_a$. Such a representation is
called a Kraus representation. The converse is standard
\cite[Theorem~2.22]{watrous2018}, but only this direct implication is
needed below.

\subsection{Implicit differentiation and complex extensions}
\begin{theorem}[Smooth implicit function theorem]\label{thm:standardift}
Let $d_{\rm par},d_{\rm var}$ be positive integers and let $G$ be a
smooth map from an open subset of
$\R^{d_{\rm par}}\times\R^{d_{\rm var}}$ to $\R^{d_{\rm var}}$.
The symbol $D$ in this theorem denotes differentiation; a subscript
specifies the variable being differentiated. Let $(p_0,y_0)$ belong
to the domain of $G$. Suppose $G(p_0,y_0)=0$ and the
partial derivative $D_yG(p_0,y_0)$ is invertible. Near $(p_0,y_0)$,
the solutions are exactly $y=g(p)$ for a unique smooth function $g$.
For every parameter direction $h$, its derivative satisfies
\[
 D_yG(p,g(p))\,Dg(p)[h]+D_pG(p,g(p))[h]=0.
\]
\end{theorem}
This is the finite-dimensional implicit function theorem; see
\cite[Appendix~B]{guilleminhaine2019}. The derivative equation follows
by differentiating $G(p,g(p))=0$. We apply it to real coordinates on
Hermitian matrix space, after checking invertibility of the specific
stationarity system.

\begin{theorem}[Principal matrix square root]\label{thm:principalroot}
Let $A$ be a complex square matrix with no eigenvalue in $(-\infty,0]$.
There is a unique square root $A^{1/2}$ whose eigenvalues have positive
real parts. This root depends holomorphically on the entries of $A$
on this open set. It respects similarity, and its eigenvalues are the
principal scalar square roots of those of $A$, with multiplicities.
\end{theorem}
See \cite[Theorem~1.29 and Sections~1.2.3, 1.7]{higham2008}.
For the holomorphic dependence, choose a positively oriented contour
$\Gamma$, or union of contours, inside $\C\setminus(-\infty,0]$
enclosing the spectrum of $A$. The matrix-function contour formula
\cite[Definition~3.3 and Theorem~3.4]{highamlin2013} gives
\[
 A^{1/2}=\frac1{2\pi\mathrm i}\int_\Gamma
       z^{1/2}(zI-A)^{-1}\,dz.
\]
One common contour works in a neighborhood of a fixed $A$; the
resolvent and hence the integral are holomorphic in its entries.
On positive definite Hermitian matrices this is the usual positive
square root. In Section~\ref{sec:regularity} we verify the spectral
hypothesis for each complex perturbation before using this theorem.

\begin{theorem}[Cauchy's derivative estimate]\label{thm:cauchy}
Let $R>0$ and $V\ge0$. If a scalar function $f$ is holomorphic on $|z|<R$ and
$|f(z)|\le V$ there, then, for every integer $j\ge0$,
\[
 |f^{(j)}(0)|\le j!V/R^j.
\]
\end{theorem}
This is \cite[Chapter~2, Corollary~4.3]{steinshakarchi2003}.
Apply Cauchy's integral formula on circles of radius $r<R$ and let
$r$ increase to $R$. We use the estimate on complex lines through a
matrix-valued parameter space; the function being bounded is scalar.

\begin{lemma}[Polarization bound]\label{lem:polarization}
Let $T$ be a symmetric $j$-linear scalar-valued form on a real normed
space, for an integer $j\ge1$, and set $P(v)=T[v,\ldots,v]$. Then
\[
 T[v_1,\ldots,v_j]=\frac1{2^j j!}
 \sum_{\varepsilon\in\{-1,1\}^j}
 \left(\prod_{i=1}^j\varepsilon_i\right)
 P\left(\sum_{i=1}^j\varepsilon_i v_i\right).
\]
Consequently its multilinear norm satisfies
\[
 \sup_{\|v_i\|\le1}|T[v_1,\ldots,v_j]|
 \le\frac{j^j}{j!}\sup_{\|v\|\le1}|P(v)|.
\]
\end{lemma}
\begin{proof}
Expand the signed sum by multilinearity. Summing each sign cancels
all terms except those containing every $v_i$ exactly once. There
are $j!$ such terms, each with sign sum $2^j$. For the bound,
$\|\sum_i\varepsilon_i v_i\|\le j$ and homogeneity of $P$ give
the displayed factor. This proves the precise form used below.
\end{proof}

\section{Quantitative regularity of the optimized value}\label{sec:regularity}

The finite walk needs one step size and one set of query accuracies
that work throughout all its possible states. The qualitative
smoothness argument does not bound them. We obtain a quantitative
bound by extending the objective to a complex neighborhood and
applying Cauchy's estimate. Relative perturbations of the source
keep the neighborhood independent of its smallest positive eigenvalue.

The first two derivatives of the weights drove the curvature argument.
Higher derivatives enter here through the finite-step remainder and
the curvature stencil. For the actual source,
\[
 c^{(3)}(t)=\frac{3u\zeta^2t}{(t^2+\zeta^2)^{5/2}},\qquad
 c^{(4)}(t)=\frac{3u\zeta^2(\zeta^2-4t^2)}{(t^2+\zeta^2)^{7/2}}.
\]
These derivatives can grow as $\zeta$ shrinks. The estimate below
controls their effect together with the changing density; it does not
assume they have a constant bound. Throughout the walk $|x_i|\le1$. The source satisfies
$0\le c(x_i)\le128$ throughout the cube. Therefore
\[
 \Omega_x(I_D)\preceq256\epsilon I_D.
\]
At a real optimizer, Lemma~\ref{lem:floor} gives
\[
 \theta S^{-1/2}\preceq40I_D,
 \qquad S\succeq\mu I_D,
 \qquad \mu=(\theta/40)^2.
\]
To see the first inequality directly, pair stationarity with $S$ to
obtain $\lambda=\mathscr F(x)-\theta\Tr\sqrt S\le\mathscr F(x)$,
drop the positive transport terms, and use
$\norm{J\widehat H(x)}\le1$.

The following are scalar budgets: $R_{\mathrm c}$ is a complex
neighborhood radius, $C_{\mathrm{val}}$ bounds the joint objective, $C_{\mathrm{der}}$ bounds its
mixed derivatives, $g$ is a density coercivity bound, and $M$ bounds
the derivatives after optimizing the density. 
Set
\begin{equation}\label{eq:regularitybudget}
\begin{gathered}
 R_{\rm c}=\min\{1,\mu,\rho,\zeta\}/10^4,\qquad
 C_{\rm val}=10^4(D+1)^2,\\
 C_{\rm der}=C_{\rm val}(10/R_{\rm c})^4,\qquad g=\theta/2,\\
 M=1+(C_{\rm der}+3C_{\rm der}^2/g)(1+C_{\rm der}/g)^4.
\end{gathered}
\end{equation}
Appendix~\ref{sec:runtime} bounds these expressions by polynomials in
$n,m$. Their definitions involve only the density and geometric scales.

\begin{lemma}[A uniform complex neighborhood]\label{lem:complex-neighborhood}
Fix a face and a real state $x$ whose live coordinates have margin at
least $\rho/2$. Let $S$ be its optimizing density, let
$\mathcal K=\operatorname{ran}\Omega_x(I_D)$, and let $b\in\R^n$
vanish on frozen coordinates with $|b_i|\le2$. Let
$\mathcal I=\{i:|x_i|<1,\ A_i\ne0\}$.
With the principal square roots, the joint objective
$L(x+zb,S+Y)$ has a holomorphic scalar extension on
$|z|<R_{\rm c}$, $\|Y\|_{\rm F}<R_{\rm c}$, where $Y$ is a full
complex matrix. Its absolute value is at most $C_{\rm val}$.
Its real mixed derivatives at $(z,Y)=(0,0)$, of orders two through
four, are bounded by $C_{\rm der}$ in the joint norm $|z|+\|Y\|_{\rm F}$.
\end{lemma}
\begin{proof}
On $\mathcal I$, the weights and probes satisfy
$c_i\ge u\rho/2$ and $p_i=\Tr(B_iS)>0$, while $S\succeq\mu I_D$.
Extend $c$ to complex arguments by the same formula, using the
principal scalar square root, and extend $\Omega$ by its same finite
sum. This extension is holomorphic on the neighborhood specified below.
The scalar
square root $\sqrt{w^2+\zeta^2}$ is holomorphic in the strip
$|\Im w|<\zeta$. On $|\Im w|\le\zeta/2$,
$|w/\sqrt{w^2+\zeta^2}|\le1$: after squaring this inequality, it follows
from
\[
 |w^2+\zeta^2|^2-|w|^4
 =2\zeta^2(\Re w)^2-2\zeta^2(\Im w)^2+\zeta^4\ge\zeta^4/2.
\]
For $|w|\le3$ this bounds $|c'(w)|$ by $7u$. For a complex scalar
displacement $|z|\le R_{\mathrm c}$ in a direction with $|b_i|\le2$,
\[
 \frac{|c(x_i+zb_i)-c(x_i)|}{c(x_i)}
 \le\frac{28|z|}{\rho}.
\]
The coefficient displacement $z$ and the full complex matrix
perturbation $Y$ are independent variables. For
$\|Y\|_{\mathrm F}\le R_{\mathrm c}$, define the full matrix
$S^Y=S+Y$, its compression $S_0^Y$ to $\mathcal K$, and the supported
matrix $M^{z,Y}=\Omega_{x+zb}(S_0^Y)$. The same finite sums define
$H(x+zb)$ and $\Omega_{x+zb}$ for complex $z$. Write
$\Delta c_i=c(x_i+zb_i)-c(x_i)$ and $\Delta p_i=\Tr(B_iY)$.
The density perturbation changes $p_i=\Tr(B_iS)$ relatively by at most $R_{\mathrm c}/\mu$.
If $\mathcal I$ is empty, the source and fidelity are zero on this
fixed face and the following source estimates are unnecessary.
Otherwise let $S_0$ be the compression of $S$ and
$M_0=\Omega_x(S_0)\succ0$ on $\mathcal K$ and
\[
 E_i=M_0^{-1/2}c_i p_iB_iM_0^{-1/2},\qquad \sum_i E_i=I.
\]
The relative source perturbation is a sum of these $E_i$ with scalar
coefficients
$(1+\Delta c_i/c_i)(1+\Delta p_i/p_i)-1$. Its operator norm is
at most the largest absolute coefficient: write the sum as a compression
of a scalar block diagonal matrix by the isometry
$w\mapsto(E_i^{1/2}w)_i$. The choices above make this norm, and the
relative compressed-density perturbation, smaller than $1/16$.

The complex extension requires the spectral hypothesis of
Theorem~\ref{thm:principalroot}. We check it using numerical ranges.
For a complex matrix $A$, its numerical range is the set of scalar
values $v^*Av$ with $\|v\|=1$. If
$A=A_0^{1/2}(I+E)A_0^{1/2}$, $A_0\succ0$, and $\|E\|<1/16$,
then every nonzero quadratic form value has argument of absolute
value less than $\alpha=\arctan(1/15)$. Indeed, after writing
$w=A_0^{1/2}v$, the real part of $w^*(I+E)w$ is at least
$(15/16)\|w\|^2$ and the imaginary part has absolute value at most
$\|w\|^2/16$. The inverse has the same bound: with $w=A^{-1}v$,
$v^*A^{-1}v=\overline{w^*Aw}$.
Apply this to $S_0^Y$ and $M^{z,Y}$. An eigenvalue $\lambda$ of
$S_0^YM^{z,Y}$ satisfies
\[
 \lambda=\frac{v^*M^{z,Y}v}{v^*(S_0^Y)^{-1}v}
 \quad\text{for some }v\ne0,
 \qquad |\arg\lambda|<2\alpha<\pi.
\]
Thus this product avoids the closed negative real axis. The same
sector argument applies to the full matrix $S^Y$.

Consequently the following scalar function is holomorphic in the
independent complex variables $(z,Y)$ on the stated neighborhood:
\[
\begin{split}
 \mathcal L_{\mathrm{ext}}(z,Y)
 ={}&\Tr\bigl((J\widehat H(x+zb))S^Y\bigr)\\
 &+2\Tr_{\mathcal K}\bigl((S_0^YM^{z,Y})^{1/2}\bigr)
   +2\theta\Tr\bigl((S^Y)^{1/2}\bigr).
\end{split}
\]
Here $\Tr_{\mathcal K}$ is the trace on the supported space, and
all other traces are full-dimensional. If $\mathcal K=\{0\}$, the
middle term is zero. For real $z$ and Hermitian $Y$ with $S+Y\succ0$,
this function equals $L(x+zb,S+Y)$. For the fidelity term this follows
by similarity of $S_0^YM^{z,Y}$ to
$(M^{z,Y})^{1/2}S_0^Y(M^{z,Y})^{1/2}$ at positive real inputs.
Theorem~\ref{thm:principalroot} therefore supplies a complex extension
of exactly the real objective used by the walk.

The bound $C_{\mathrm{val}}$ dominates the absolute joint density objective on this
complex ball. For example, $\norm{S^Y}\le2$, the center has norm at
most $6$, and
\[
 \norm{M^{z,Y}}\le6400m\le3200D
\]
follows by summing the conservative bound $|c(x_i+zb_i)|\le1600$,
$|\Tr(B_iS^Y)|\le4\norm{A_i}$, and
$\sum_i\norm{A_i}^2\le m$. The trace of the principal square root of
$S_0^YM^{z,Y}$ is bounded by $D\sqrt{\norm{S^Y}\norm{M^{z,Y}}}$.
Theorem~\ref{thm:cauchy} and Lemma~\ref{lem:polarization} bound all joint derivatives
of orders two through four by $C_{\mathrm{der}}$. More explicitly, equip a joint
direction $(s,Y)$ with norm $|s|+\norm{Y}_{\mathrm F}$. The fourth
derivative in a unit direction has bound $4!C_{\mathrm{val}}/R_{\mathrm c}^4$.
Lemma~\ref{lem:polarization} gives a factor at most $4^4/4!$ in
passing from diagonal to mixed evaluations. The factor $10^4$ in $C_{\mathrm{der}}$ exceeds
this allowance and the smaller-order allowances. We use the same
$C_{\mathrm{der}}$ for all orders two through four.

\end{proof}

\begin{lemma}[Uniform derivatives of the optimized potential]\label{lem:fourth}
At every prepared state, along every fixed live direction with
$|b_i|\le2$, the second through fourth derivatives of
$t\mapsto\mathscr F(x+tb)$ have absolute value at most $M$ for
$|t|\le\rho/16$.
\end{lemma}
\begin{proof}
Every point on the segment has live margin at least $\rho/2$.
Lemma~\ref{lem:complex-neighborhood} bounds the joint derivatives there.
It remains to control the change in the maximizing density. At a
point of this segment, write $S$ for its optimizer.
On the real full trace-zero density tangent, write $s_a$ for the
eigenvalues of $S$. The negative regularizer
Hessian has eigenbasis coefficients
\[
 \frac{\theta}{\sqrt{s_a}\sqrt{s_b}(\sqrt{s_a}+\sqrt{s_b})}
 \ge g.
\]
Fidelity concavity only increases coercivity. Write
$\ell(t,T)=L(x+tb,T)$ for the real joint objective, let $S(t)$ be its
optimizing full density, and set $\varphi(t)=\mathscr F(x+tb)$.
The symbols $S'$ and $S''$ below are derivatives of $S(t)$.
For $j\ge1$, $\ell^{(j)}$ denotes the $j$th joint derivative evaluated
as a multilinear form on pairs consisting of a scalar displacement
and a full Hermitian variation, at $(t,S(t))$.
Subscripts $t$ and $S$ denote differentiation in the scalar and
density variables, respectively; all density differentials in the
following calculation are restricted to trace-zero Hermitian variations.
Density stationarity on that tangent says $\ell_S=0$.
Its first derivative is
\[
 \ell_{SS}S'+\ell_{tS}=0,
 \qquad \ell^{(2)}[(1,S'),(0,Y)]=0\quad(\Tr Y=0).
\]
The first equation is an equality of linear functionals on that
tangent. Coercivity gives $\|S'\|_{\mathrm F}\le C_{\mathrm{der}}/g$.
Differentiating again yields, for every trace-zero $Y$,
\[
 \ell^{(2)}[(0,S''),(0,Y)]
 +\ell^{(3)}[(1,S'),(1,S'),(0,Y)]=0.
\]
The mixed derivative bound and coercivity therefore give
\[
 \|S''\|_{\mathrm F}
 \le(C_{\mathrm{der}}/g)(1+C_{\mathrm{der}}/g)^2.
\]
The first differentiated stationarity identity above holds for every
trace-zero variation, including $S''$.
It cancels the terms containing $S''$ when differentiating
the second envelope derivative. For the joint direction $w=(1,S')$,
the resulting envelope identities are
\[
 \varphi^{(3)}(t)=\ell^{(3)}[w,w,w],
\]
\[
 \varphi^{(4)}(t)
 =\ell^{(4)}[w,w,w,w]+3\ell^{(3)}[w,w,(0,S'')].
\]
These bounds give $M$; they also bound the second and third derivatives
by $M$. Theorem~\ref{thm:standardift} justifies the differentiated
optimizer throughout the fixed face.
\end{proof}

\section{Scalar realization and interior bounds for the SDP}
\label{sec:sdpdetails}
The value identity in Section~\ref{sec:sdp} determines the queried
optimization. Here we give its scalar encoding and quantitative bounds
for the optional positive-source perturbation. These estimates also
show that a degenerating source introduces no extra input parameter
into the construction.

\begin{lemma}[Boundedness and a strictly feasible realization]\label{lem:slater}
For $V=4u\epsilon$ the source-perturbed SDP has the bounds
\eqref{eq:O10}--\eqref{eq:O13}. In particular its reciprocal inner
radius is polynomial in $D$ and the reciprocal value accuracy.
\end{lemma}
\begin{proof}
Before the estimates, note that
$\norm{\Omega_x(Y)}_{\mathrm F}\le L_\Omega\norm Y_{\mathrm F}$, where
\[
 L_\Omega=2\sum_i c_i\norm{A_i}^2
 \le4u\epsilon m=2u\epsilon D\le2uD.
\]
This follows termwise from the Kraus representation and
$\sum_i\norm{A_i}=m$.
Use the real Hilbert norm
\(\|(S,W,Y)\|^2=\|S\|_{\rm F}^2+\|W\|_{\rm F}^2+\|Y\|_{\rm F}^2\)
on the affine hyperplane \(\operatorname{Tr}S=1\).
Every feasible point of the regularized SDP satisfies
\begin{equation}\label{eq:O10}
 \|S\|_{\rm F}\le1,\qquad \|Y\|_{\rm F}\le1,\qquad
 \|W\|_{\rm F}\le\sqrt{V+\lambda}.
\end{equation}
The first bound uses \(S\succeq0,\operatorname{Tr}S=1\).
The second follows from \(Y^2\preceq S\).
For the third, \eqref{eq:O5}'s factorization gives
\(W=S^{1/2}C(\Omega_x(S)+\lambda I)^{1/2}\), with \(\|C\|\le1\).
Since \(0\preceq S\preceq I\), positivity and \eqref{eq:kraus} imply
\(\Omega_x(S)+\lambda I\preceq(V+\lambda)I\); use
\(\|AB\|_{\rm F}\le\|A\|_{\rm F}\|B\|\).
Thus the feasible body has a polynomial outer radius. For example,
a radius \(R=2\sqrt{2+V+\lambda}+1\) about the point below suffices.

An explicit strict feasible point is the tuple $(S_0,W_0,Y_0)$ below;
these subscripts denote this reference point, not source compression:
\begin{equation}\label{eq:O11}
 S_0=I_D/D,\qquad W_0=0,\qquad Y_0=I_D/(2\sqrt D).
\end{equation}
The first block is block diagonal and has minimum eigenvalue at least
\(\min\{1/D,\lambda\}\). For the second block, each scalar
\(2\times2\) sector has determinant \(3/(4D)\) and trace
\(1+1/D\); its minimum eigenvalue is at least
\(3/[4(D+1)]\ge1/(4D)\). Also \(Y_0\succeq I_D/(2\sqrt D)\).

Set
\begin{equation}\label{eq:O12}
 \sigma=\min\{\lambda,1/(4D)\},\qquad
 r=\frac{\sigma}{4(3+2uD)}.
\end{equation}
Every affine perturbation of \((S_0,W_0,Y_0)\) of Hilbert norm at
most \(r\) remains feasible. Indeed the first block changes in
operator norm by at most \((3+L_\Omega)r\le(3+2uD)r\), the second
by at most \(3r\), and \(Y\) by at most \(r\). These are smaller
than the corresponding positive margins at \eqref{eq:O11}. The trace equality
is retained by taking the ball within its affine hyperplane.
Under \eqref{eq:O9}, \(\sigma=\lambda\), so \(r^{-1}\) is polynomial in
\(D,u,\nu^{-1}\), without any source inverse eigenvalue.

The objective has Hilbert norm at most
\begin{equation}\label{eq:O13}
 L_{\rm obj}=(\|M_x\|+2+2\theta)\sqrt D.
\end{equation}
Parseval and \(|x_i|\le1\) give \(\|M_x\|\le1\).
\end{proof}
\begin{lemma}[Materializing a value query]\label{lem:materialize}
All coefficients of \eqref{eq:O7} are computed from $((v_i)_{i=1}^n,x,\theta)$
in polynomially many real arithmetic operations. A direct scalar-data
bound is
\begin{equation}\label{eq:sdpdatasize}
 L_m=6400m^4+16m^2+2.
\end{equation}
\end{lemma}
\begin{proof}
Represent a Hermitian density by its real diagonal entries and the
real and imaginary parts of its strict upper triangle. Eliminate one
diagonal entry using $\Tr S=1$. Use the same Hermitian chart for $Y$
and two real coordinates for every entry of $W$. This gives
$(D^2-1)+D^2+2D^2=4D^2-1$ variables. Realify a Hermitian block $Q$ as
$\left(\begin{smallmatrix}\Re Q&-\Im Q\\\Im Q&\Re Q\end{smallmatrix}\right)$;
its quadratic form is nonnegative exactly when that of $Q$ is.
The block orders are $4D,4D,2D$, with total $10D=20m$.
Store one constant and $16m^2-1$ coefficient matrices, and the same
number of objective coefficients. Even storing each entire pencil,
the scalar count is at most
$(16m^2)(20m)^2+16m^2+2=L_m$.
For each chart matrix $E$, compute
$\Omega_x(E)=\sum_i c(x_i)\Tr(B_iE)B_i$ by entrywise products and
sums. Its construction requires polynomial work in $n,m$; it uses
only scalar square roots for $c(x_i)$ and no source inverse. The center and
objective coefficients are obtained by the same finite sums.
\end{proof}

\section{Polynomial parameters and arithmetic work}
\label{sec:runtime}

The deterministic signing argument is complete. We now count the
execution that realizes it, including all rejected local candidates
and both potential reports in a curvature movement.
The explicit powers below certify polynomial dependence; no effort is
made to optimize them.

\begin{lemma}[Polynomial derivative and precision budgets]
\label{lem:polyparameters}
Let $s=10^8(n+m+1)$ and $B_*=s^{160}$. The quantity $M$ defined in
Section~\ref{sec:regularity} satisfies $1<M\le B_*$. The actual
numerical parameters obey
\begin{equation}\label{eq:inverseparameters}
 \begin{gathered}
 t^{-2}\le s^{166},\qquad h^{-2}\le s^{165},\qquad
 \nu_{\rm H}^{-1}\le s^{171},\qquad
 \nu_{\rm loc}^{-1}\le s^{169},\\
 T\le s^{168},\qquad \delta^{-1}\le s.
 \end{gathered}
\end{equation}
In particular the walk makes at most $s^{168}$ movements.
\end{lemma}
\begin{proof}
All inequalities below concern explicitly defined arithmetic quantities.
Parseval gives $D\le2n$, $\delta^{-1}\le\sqrt n$, and
$\theta^{-1}\le\sqrt{2n}$. The geometric definitions give
\[
 \rho^{-1}=100n^2,\quad a^{-1}=1600n^2,\quad
 \zeta^{-1}=16000n^2,\quad \mu^{-1}=1600\theta^{-2}\le3200n.
\]
Consequently
\[
 R_{\rm c}^{-1}\le s^4,\qquad C_{\mathrm{val}}\le s^3,\qquad
 C_{\mathrm{der}}=C_{\mathrm{val}}(10/R_{\rm c})^4\le s^{23},\qquad g^{-1}\le s^2.
\]
For example each reciprocal inside the definition of $R_{\rm c}$ is
at most $s^3$, and its extra factor $10^4$ is at most $s$.
It follows that
\[
 C_{\mathrm{der}}+3C_{\mathrm{der}}^2/g\le s^{50},\qquad 1+C_{\mathrm{der}}/g\le s^{26},\qquad
 M\le1+s^{154}\le s^{155}<B_*.
\]
Also $\beta^{-1}=100n\delta^{-1}\le s^3$.
The reciprocal square of a minimum is a maximum, so
\[
 \begin{split}
 t^{-2}&=\max\{256\rho^{-2},128nM\beta^{-1}\}
           \le s^{166},\\
 h^{-2}&=\max\{256\rho^{-2},4M\beta^{-1}\}
           \le s^{165}.
 \end{split}
\]
The slack in these exponents includes every numerical constant.
Multiplication by $128n\beta^{-1}$ gives the stated inverse bound for
$\nu_{\rm H}$; multiplication by $8\beta^{-1}$ gives the one for
$\nu_{\rm loc}$. Finally
$T\le16nh^{-2}+1\le s^{168}$.
\end{proof}

\subsection{Work of the tangent frame and spectral calls}
A curvature branch first constructs the explicit tangent frame from
Lemma~\ref{lem:tangentframe}. At $z=0$ this is the identity. Otherwise
normalization, the sign comparison, and the Householder entries use
only scalar expressions. The zero test can be made by computing
$\|z\|_2$ and comparing it with zero. All divisions occur on their
proved nonzero branches.

Forming $Q_x^T\widetilde KQ_x$ requires two matrix multiplications of
order at most $n$. Its exact eigendecomposition is one call to the
spectral primitive. Scanning the diagonal selects the least signed
eigenvalue, and multiplication by $Q_x$ returns the physical unit
direction. These operations have polynomial work in $n$. There are
at most $T$ such spectral calls; the discrepancy guarantee follows
from the nonincreasing account, so the algorithm returns its final
coefficient array directly.

\begin{theorem}[Complete arithmetic accounting]\label{thm:runtime}
For a fixed solver in Model~\ref{model:solver}, a polynomial in $n,m$
bounds the arithmetic work of Algorithm~\ref{alg:trial}.
\end{theorem}
\begin{proof}
First form the atoms and their norms. Rank one gives
$\norm{A_i}=\norm{v_i}^2$, so the maximum $\epsilon$, $\delta$,
$\theta$, the geometric scales, and the derivative and precision
budgets are computed by finite sums, products, comparisons, fixed
integer powers, and scalar square roots. Positivity has been proved
before every reciprocal. The integer ceiling for $T$ can be formed
by incrementing and comparing against its defining real expression;
the polynomial cap bounds this work too.

There are at most $T$ movements and $n$ boundary roundings. Scanning
all labels to prepare a state costs polynomial work. An iteration
uses at most $n+1$ outward-test value calls, at most
$4n^2+3n+1$ Hessian value calls, and two further value calls to select
the sign of a curvature movement. All queried states stay in the cube
and keep frozen coordinates fixed. For each query, construct the
scalar SDP arrays of Lemma~\ref{lem:materialize}; their size is at most
$L_m=6400m^4+16m^2+2$, and their construction is polynomial in $n,m$.
The source coefficients use only the scalar square roots already
allowed in the model.

With $s=10^8(n+m+1)$ from Lemma~\ref{lem:polyparameters}, each requested
inverse accuracy is at most $s^{171}$. One solver call therefore costs
at most $A_*(L_m+1+s^{171})^{b_*}$ operations. The bound covers rejected
outward candidates and both signs of a curvature movement. The frame,
compressed EVD, and scalar array updates have the polynomial costs
just described. The algorithm follows one deterministic sequence of
states and stops when it reaches a vertex. Equivalently, the runtime
account can pad this sequence to $T$ iterations with absorbing terminal
steps; this has the same output and still has polynomial work.

For $m=0$ the all-positive signing is returned before any inverse
scale is computed. Finally, replacing each maximum by a sum and each
ceiling by its polynomial majorant produces a finite sum and product
of fixed polynomials. Its coefficients and degree depend on the fixed
solver bounds, while its variables are only $n$ and $m$.
\end{proof}

\section{The marginal soft edge of a free semicircular operator}\label{sec:freeinterpretation}

We identify the finite-dimensional potential with a soft spectral
edge of an operator-valued semicircular perturbation. This supplies
the free-probability interpretation used in the motivation: the free
operator is eliminated through Lehner's variational formula, and the
regularizer remains on the matrix marginal of the test state. All
operators and states needed for this identification are constructed
below.

\subsection{The free operator and the marginal soft edge}
Fix a state $x\in[-1,1]^n$ with $m>0$, and abbreviate its signed center by
$A_{\rm c}=J\widehat H(x)$. In particular $D>0$ and $\theta>0$.
The source weights $c_i=c(x_i)$ and Pauli
matrices $\sigma_\alpha$ retain their definitions from
Definition~\ref{def:potential} and Lemma~\ref{lem:kraus}.
Let $\mathcal I=[n]\times\{0,1,2,3\}$, and for
$\gamma=(i,\alpha)\in\mathcal I$ define the Hermitian matrix
\[
 T_\gamma=\sqrt{c_i/2}\,(\sigma_\alpha\otimes A_i).
\]
Equation~\eqref{eq:kraus} says exactly that
$\Omega_x(V)=\sum_{\gamma\in\mathcal I}T_\gamma V T_\gamma$
for every $D\times D$ matrix $V$.

Let $\mathcal H_{\rm f}$ be the Hilbert space with orthonormal basis
$e_w$ indexed by all finite words $w$ in the alphabet $\mathcal I$,
including the empty word. This is full Fock space. Write $\omega$ for
the empty-word basis vector, and $P_\omega$ for the orthogonal projection
onto its span. For each letter $\gamma$, the left creation
operator $\ell_\gamma$ is defined by
$\ell_\gamma e_w=e_{\gamma w}$; its adjoint removes an initial
$\gamma$, and gives zero otherwise. In particular
$\ell_\gamma^*\ell_\nu=I$ for $\gamma=\nu$ and is zero for
distinct letters $\gamma,\nu\in\mathcal I$.
The operators $s_\gamma=\ell_\gamma+\ell_\gamma^*$ form the standard
variance-one free semicircular family in the vacuum state
$\tau(B)=\langle\omega,B\omega\rangle$, defined for bounded operators
$B$ on $\mathcal H_{\rm f}$. Here this description refers
to the explicit creation-operator model just given. Set
\begin{equation}\label{eq:free:model}
 X_x=A_{\rm c}\otimes I+
            \sum_{\gamma\in\mathcal I}T_\gamma\otimes s_\gamma
       \quad\hbox{on }\C^D\otimes\mathcal H_{\rm f}.
\end{equation}
It is bounded and self-adjoint, since the sum is finite and
$\norm{s_\gamma}\le2$. For an operator block matrix $B=(B_{ab})_{a,b=1}^D$,
define $({\rm id}\otimes\tau)(B)=(\tau(B_{ab}))_{a,b=1}^D$.
If $Y_x=X_x-A_{\rm c}\otimes I$, its
operator-valued covariance is
\[
 ({\rm id}\otimes\tau)\bigl(Y_x(V\otimes I)Y_x\bigr)=\Omega_x(V).
\]
Indeed $\tau(s_\gamma s_\nu)$ is one when the letters agree and zero
otherwise, so expansion gives the displayed Kraus sum.

Write $\mathcal D_D=\{S\succeq0:\Tr S=1\}$ for the finite density
space. A joint density operator $\Gamma$ is a positive trace-class
operator on $\C^D\otimes\mathcal H_{\rm f}$ with $\Tr\Gamma=1$.
Its matrix marginal $S_\Gamma=\Tr_{\mathcal H_{\rm f}}\Gamma$ is
characterized by
$\Tr(VS_\Gamma)=\Tr((V\otimes I)\Gamma)$ for every matrix $V$;
it belongs to $\mathcal D_D$.
For a Hermitian $D\times D$ matrix $V$, and a bounded self-adjoint
operator $Y$ on the joint space, define
\begin{align}
 f_\theta(V)&=\max_{S\in\mathcal D_D}
       \{\Tr(VS)+2\theta\Tr\sqrt S\},\label{eq:free:matrixsoft}\\
 \mathcal E_\theta(Y)&=\sup_{\Gamma\succeq0,\,\Tr\Gamma=1}
       \{\Tr(\Gamma Y)+2\theta\Tr\sqrt{S_\Gamma}\}.
       \label{eq:free:marginalsoft}
\end{align}
The second supremum ranges over joint density operators as just
defined. The regularizer is $\theta$ times the Tsallis--$1/2$ entropy
$2(\Tr\sqrt S-1)$, plus the constant $2\theta$.
It acts on the finite matrix marginal of the \emph{test state}.
Taking the free expectation of the operator itself would instead give
$({\rm id}\otimes\tau)(X_x)=A_{\rm c}$ and lose the covariance.

\begin{theorem}[Exact marginal-edge identity]\label{thm:free:identity}
With the infimum over full $D\times D$ positive definite matrices $Z$,
at every state, including states with a singular or zero source,
\begin{equation}\label{eq:free:identity}
 \mathscr F(x)
 =\inf_{Z\succ0}f_\theta(A_{\rm c}+Z^{-1}+\Omega_x(Z))
 =\mathcal E_\theta(X_x).
\end{equation}
The transport infimum and the joint-state
supremum need not be attained. The density maximum defining
$\mathscr F(x)$ is attained.
\end{theorem}
The proof occupies the next two subsections. For comparison with the
ordinary edge, write $\operatorname{spec}(Y)$ for the spectrum of a
bounded self-adjoint operator $Y$ and define
$\lambda_+(Y)=\sup\operatorname{spec}(Y)$. Since
$1\le\Tr\sqrt S\le\sqrt D$ on $\mathcal D_D$,
\[
 \lambda_+(X_x)+2\theta\le\mathcal E_\theta(X_x)
              \le\lambda_+(X_x)+2\theta\sqrt D.
\]
The lower bound uses unit-vector states approaching the upper spectral
edge. Thus the regularization error depends only on the finite matrix
dimension, despite the infinite-dimensional auxiliary space.

\subsection{Lehner's formula and the matrix Dyson equation}
Keep the Hermitian factors $T_\gamma$ fixed, but allow an arbitrary
Hermitian center $V$. Define
$Y_V=V\otimes I+\sum_\gamma T_\gamma\otimes s_\gamma$.
We prove the upper-edge form of Lehner's semicircular formula
\cite[Corollary~1.5]{lehner1999}:
\begin{equation}\label{eq:free:lehner}
 \lambda_+(Y_V)=\inf_{Z\succ0}
       \lambda_{\max}(V+Z^{-1}+\Omega_x(Z)).
\end{equation}
Lehner states a norm formula for a positive center; the following
argument gives the upper-edge version directly for every Hermitian
center.
The finite-dimensional infimum on the right is equivalent to an SDP.
Introducing a scalar $t$ gives the formulation
\[
 \inf_{t\in\R,\ Z=Z^*}
 \left\{t:
 \begin{pmatrix}
 tI_D-V-\Omega_x(Z)&I_D\\
 I_D&Z
 \end{pmatrix}\succeq0\right\}.
\]
Block positivity first gives $Z\succeq0$. If $Zv=0$, positivity
forces the full block matrix to annihilate $(0,v)$, whose first
component after multiplication is $v$. Thus $v=0$, so $Z\succ0$.
Its Schur complement is
$tI_D-V-\Omega_x(Z)-Z^{-1}\succeq0$. Since $\Omega_x$ is linear,
the block is affine in $(t,Z)$; see
\cite[Section~4.6.2 and Appendix~A.5.5]{boydvandenberghe2004}
for this standard SDP reformulation. The corresponding
Tsallis--$1/2$ regularized formulation is the SDP already derived
in Propositions~\ref{prop:sdp} and~\ref{prop:dual}. These
reformulations concern the finite-dimensional optimization problems;
their identification with free spectral edges is proved here.

Put $\mathcal C=\sum_\gamma T_\gamma\otimes\ell_\gamma$.
For $Z\succ0$, expanding a positive square gives
\begin{align*}
 0&\preceq
 \bigl[Z^{-1/2}\otimes I-(Z^{1/2}\otimes I)\mathcal C\bigr]^*
 \bigl[Z^{-1/2}\otimes I-(Z^{1/2}\otimes I)\mathcal C\bigr]\\
 &=\bigl(Z^{-1}+\Omega_x(Z)\bigr)\otimes I
                   -(\mathcal C+\mathcal C^*).
\end{align*}
The last equality uses the relation between creation operators stated
above. Hence $Y_V\preceq(V+Z^{-1}+\Omega_x(Z))\otimes I$,
which proves one direction of \eqref{eq:free:lehner}.

For the reverse direction, take a real $t>\lambda_+(Y_V)$.
Let $\iota_\omega:\C^D\to\C^D\otimes\mathcal H_{\rm f}$ send
$v$ to $v\otimes\omega$, and set
$G_V(t)=\iota_\omega^*(tI-Y_V)^{-1}\iota_\omega$.
The operator $tI-Y_V$ is bounded and strictly positive, so its inverse
is bounded and strictly positive and $G_V(t)\succ0$.
Partition all nonempty words by their \emph{last} letter. Each branch
ending in $\gamma$ is a copy of Fock space under $w\mapsto w\gamma$.
Compression of $Y_V$ to that branch is another copy of $Y_V$, since
creation prepends a letter. The only connection from the empty word
to that branch has coefficient $T_\gamma$ and joins its root.
In the decomposition into the empty word and these branches,
$tI-Y_V$ therefore has root block $tI_D-V$, branch block
$\bigoplus_\gamma(tI-Y_V)$, and off-diagonal blocks determined by
$-T_\gamma\iota_\omega^*$.
Schur elimination yields
\begin{equation}\label{eq:free:dyson}
 G_V(t)^{-1}=tI_D-V-\Omega_x(G_V(t)),\qquad
 V+G_V(t)^{-1}+\Omega_x(G_V(t))=tI_D.
\end{equation}
This Schur elimination is valid for bounded operators: the branch
block has a bounded inverse, and multiplication by the usual bounded
invertible triangular block matrices diagonalizes the block operator.
Its inverse root block is the inverse of its Schur complement.
Choosing $Z=G_V(t)$ and letting $t$ decrease to $\lambda_+(Y_V)$
proves \eqref{eq:free:lehner}, without assuming a limiting positive
transport at the edge.

Equation~\eqref{eq:free:dyson} is the matrix Dyson equation outside the
upper spectrum. The common convention uses the opposite resolvent
sign: with $M_V(t)=-G_V(t)$ it reads
$-M_V(t)^{-1}=tI_D-V+\Omega_x(M_V(t))$, as in
\cite[Section~2.1, equation~(2.2)]{aek2019}.
The Fock-space Schur complement derives this equation directly from
the resolvent. The regularized optimization will select a different
transport through its density stationarity equation.

\subsection{Marginal duality and the regularized transport}
We first prove a finite-dimensional duality statement despite the
infinite joint space. For every bounded self-adjoint $Y$,
\begin{equation}\label{eq:free:marginaldual}
 \mathcal E_\theta(Y)=\inf_{V=V^*}
       \{\lambda_+(Y-V\otimes I)+f_\theta(V)\}.
\end{equation}
For every joint density $\Gamma$, split its objective into
\[
 \Tr(\Gamma(Y-V\otimes I))
            +\Tr(VS_\Gamma)+2\theta\Tr\sqrt{S_\Gamma}.
\]
Each term is bounded by the corresponding term in
\eqref{eq:free:marginaldual}. This proves the upper bound for
$\mathcal E_\theta(Y)$.

For equality let $\mathcal Q_Y$ be the closure of the pairs
$(S_\Gamma,\Tr(\Gamma Y))$ in the finite-dimensional real space of
Hermitian matrices and scalars. This set is convex and compact: its
first component lies in $\mathcal D_D$, and its second lies in
$[-\norm Y,\norm Y]$. For every Hermitian matrix $T$ its support
function in direction $(T,1)$ is
\begin{equation}\label{eq:free:support}
 \max_{(S,e)\in\mathcal Q_Y}\{e+\Tr(TS)\}
                       =\lambda_+(Y+T\otimes I).
\end{equation}
Indeed, expectations in unit-vector states approach the upper edge,
and taking closure preserves the supremum of this continuous linear
function. Choose $(S_*,e_*)\in\mathcal Q_Y$ maximizing
$e+2\theta\Tr\sqrt S$. Its value is $\mathcal E_\theta(Y)$ by
continuity. Moreover $S_*\succ0$: mix this pair with that of the
product state $(I_D/D)\otimes P_\omega$.
For mixture weight $t>0$, a zero eigenvalue of $S_*$ gives a positive
gain of order $\sqrt t$ in the regularizer, while the energy and the
positive-eigenvalue contributions change by $O(t)$.
This contradicts maximality for small positive $t$.

Put $T_* =\theta S_*^{-1/2}$, the gradient of the regularizer.
Differentiating along segments in $\mathcal Q_Y$ gives
$e+\Tr(T_*S)\le e_*+\Tr(T_*S_*)$ for all $(S,e)\in\mathcal Q_Y$.
The supporting-plane inequality for the concave regularizer also gives
$f_\theta(-T_*)=2\theta\Tr\sqrt{S_*}-\Tr(T_*S_*)$.
Taking $V=-T_*$ in \eqref{eq:free:marginaldual} and using
\eqref{eq:free:support} proves equality, including when the joint-state
supremum is attained only in the closure of its marginal-energy pairs.

For any Hermitian matrix $W$, one further elementary identity is
\begin{equation}\label{eq:free:infcenter}
 \inf_{V=V^*}\{\lambda_{\max}(W-V)+f_\theta(V)\}=f_\theta(W).
\end{equation}
The inequality $W\preceq V+\lambda_{\max}(W-V)I_D$ gives the lower
bound by testing every density in the definition of $f_\theta$.
Taking $V=W$ gives equality. Apply \eqref{eq:free:lehner} to
$X_x-V\otimes I$, substitute into \eqref{eq:free:marginaldual},
and commute the two infima, which range over a Cartesian product.
Then \eqref{eq:free:infcenter} yields
\begin{equation}\label{eq:free:softtransport}
 \mathcal E_\theta(X_x)
       =\inf_{Z\succ0}f_\theta(A_{\rm c}+Z^{-1}+\Omega_x(Z)).
\end{equation}
The reduction uses only finite-dimensional marginal duality and the
commutation of two infima over independent variables.

It remains to match this expression to the implemented potential.
For every Hermitian $D\times D$ matrix $W$, with $Q$ ranging over
positive definite matrices of the same size,
\begin{equation}\label{eq:free:softdual}
 f_\theta(W)=\inf_{Q\succ0}
       \{\lambda_{\max}(W+Q)+\theta^2\Tr Q^{-1}\}.
\end{equation}
The transport inequality in the proof of Proposition~\ref{prop:dual}
gives
$2\theta\Tr\sqrt S\le\Tr(SQ)+\theta^2\Tr Q^{-1}$, proving
one direction. A maximizing density for $f_\theta(W)$ exists by
compactness and is positive definite by the same mixture argument.
Denote it by $S_W$. Stationarity on the trace-one affine space gives
$W+\theta S_W^{-1/2}=\lambda_W I_D$ for a real scalar $\lambda_W$.
Choose $Q=\theta S_W^{-1/2}$. The right side of
\eqref{eq:free:softdual} then equals
$\lambda_W+\theta\Tr\sqrt{S_W}=f_\theta(W)$.
Substituting \eqref{eq:free:softdual} into
\eqref{eq:free:softtransport} gives exactly the already-proved dual
\eqref{eq:transportdual}. Proposition~\ref{prop:dual}, including its
singular-source argument, therefore proves Theorem~\ref{thm:free:identity}.

The identity clarifies both the connection to the MDE and the change
introduced by regularization. The ordinary resolvent $G_V(t)$ obeys
\eqref{eq:free:dyson}. The variational transport obeys
\eqref{eq:9}, with the additional matrix $\theta S^{-1/2}$ and
trace-one density constraint. Its response is therefore a stability
problem for this regularized system. The analysis in
Sections~\ref{sec:curvature}--\ref{sec:sm:curvature} bounds the
particular perturbations generated by the walk, through a joint-kernel
projection. This is structurally related to the study of linearized
Dyson equations near edges \cite{aeks2020,alt2020}, while its
operators and estimates are specific to the density potential.

\begingroup
\small

\begin{thebibliography}{10}

\bibitem{aek2019}
Oskari~H. Ajanki, L{\'a}szl{\'o} Erd{\H{o}}s, and Torben Kr{\"u}ger.
\newblock Stability of the matrix dyson equation and random matrices with
  correlations.
\newblock {\em Probability Theory and Related Fields}, 173:293--373, 2019.
\newblock \href {https://doi.org/10.1007/s00440-018-0835-z}
  {\path{doi:10.1007/s00440-018-0835-z}}.

\bibitem{allenZhuLiaoOrecchia2015}
Zeyuan Allen-Zhu, Zhenyu Liao, and Lorenzo Orecchia.
\newblock Spectral sparsification and regret minimization beyond matrix
  multiplicative updates, 2015.
\newblock STOC 2015; full version.
\newblock URL: \url{https://arxiv.org/abs/1506.04838}, \href
  {https://arxiv.org/abs/1506.04838} {\path{arXiv:1506.04838}}.

\bibitem{alt2020}
Johannes Alt, L{\'a}szl{\'o} Erd{\H{o}}s, and Torben Kr{\"u}ger.
\newblock The dyson equation with linear self-energy: Spectral bands, edges and
  cusps.
\newblock {\em Documenta Mathematica}, 25:1421--1539, 2020.
\newblock URL: \url{https://arxiv.org/abs/1804.07752}.

\bibitem{aeks2020}
Johannes Alt, L{\'a}szl{\'o} Erd{\H{o}}s, Torben Kr{\"u}ger, and Dominik
  Schr{\"o}der.
\newblock Correlated random matrices: Band rigidity and edge universality.
\newblock {\em Annals of Probability}, 48(2):963--1001, 2020.
\newblock \href {https://doi.org/10.1214/19-AOP1379}
  {\path{doi:10.1214/19-AOP1379}}.

\bibitem{anari2017}
Nima Anari, Shayan Oveis~Gharan, Amin Saberi, and Nikhil Srivastava.
\newblock Approximating the largest root and applications to interlacing
  families, 2017.
\newblock URL: \url{https://arxiv.org/abs/1704.03892}, \href
  {https://arxiv.org/abs/1704.03892} {\path{arXiv:1704.03892}}.

\bibitem{bbvh2023}
Afonso~S. Bandeira, March Boedihardjo, and Ramon van Handel.
\newblock Matrix concentration inequalities and free probability.
\newblock {\em Inventiones mathematicae}, 234:419--487, 2023.
\newblock \href {https://doi.org/10.1007/s00222-023-01204-6}
  {\path{doi:10.1007/s00222-023-01204-6}}.

\bibitem{bhatia2007}
Rajendra Bhatia.
\newblock {\em Positive Definite Matrices}.
\newblock Princeton University Press, 2007.
\newblock URL: \url{https://assets.press.princeton.edu/chapters/s8445.pdf}.

\bibitem{boydvandenberghe2004}
Stephen Boyd and Lieven Vandenberghe.
\newblock {\em Convex Optimization}.
\newblock Cambridge University Press, 2004.
\newblock URL: \url{https://web.stanford.edu/~boyd/cvxbook/bv_cvxbook.pdf}.

\bibitem{erdos2019}
L{\'a}szl{\'o} Erd{\H{o}}s.
\newblock The matrix dyson equation and its applications for random matrices,
  2019.
\newblock URL: \url{https://arxiv.org/abs/1903.10060}, \href
  {https://arxiv.org/abs/1903.10060} {\path{arXiv:1903.10060}}.

\bibitem{guilleminhaine2019}
Victor Guillemin and Peter Haine.
\newblock {\em Differential Forms}.
\newblock World Scientific, 2019.
\newblock URL:
  \url{https://math.mit.edu/classes/18.952/2018SP/files/18.952_book.pdf}, \href
  {https://doi.org/10.1142/11058} {\path{doi:10.1142/11058}}.

\bibitem{higham2008}
Nicholas~J. Higham.
\newblock {\em Functions of Matrices: Theory and Computation}.
\newblock SIAM, 2008.
\newblock \href {https://doi.org/10.1137/1.9780898717778}
  {\path{doi:10.1137/1.9780898717778}}.

\bibitem{highamlin2013}
Nicholas~J. Higham and Lijing Lin.
\newblock Matrix functions: A short course.
\newblock Technical Report 2013.73, Manchester Institute for Mathematical
  Sciences, 2013.
\newblock URL:
  \url{https://eprints.maths.manchester.ac.uk/2067/1/covered/MIMS_ep2013_73.pdf}.

\bibitem{jourdan2023}
Ben Jourdan, Peter Macgregor, and He~Sun.
\newblock Is the algorithmic {Kadison--Singer} problem hard?
\newblock In {\em 34th International Symposium on Algorithms and Computation
  (ISAAC 2023)}, volume 283 of {\em Leibniz International Proceedings in
  Informatics}, pages 43:1--43:18, 2023.
\newblock \href {https://doi.org/10.4230/LIPIcs.ISAAC.2023.43}
  {\path{doi:10.4230/LIPIcs.ISAAC.2023.43}}.

\bibitem{jourdanMacgregorSun2024}
Ben Jourdan, Peter Macgregor, and He~Sun.
\newblock Polynomial-time algorithms for weaver's discrepancy problem in a
  dense regime, 2024.
\newblock URL: \url{https://arxiv.org/abs/2402.08545}, \href
  {https://arxiv.org/abs/2402.08545} {\path{arXiv:2402.08545}}.

\bibitem{kadison1959}
Richard~V. Kadison and Isadore~M. Singer.
\newblock Extensions of pure states.
\newblock {\em American Journal of Mathematics}, 81(2):383--400, 1959.
\newblock \href {https://doi.org/10.2307/2372748} {\path{doi:10.2307/2372748}}.

\bibitem{kathuria2026ms}
Tarun Kathuria.
\newblock {A Walk From Free Probability to Matrix Discrepancy I: Matrix
  Spencer}.
\newblock Companion paper, 2026.

\bibitem{kathuria2026higherRank}
Tarun Kathuria.
\newblock {A Walk From Free Probability to Matrix Discrepancy III: Higher Rank
  Kadison--Singer and Spectrally Thin Trees}.
\newblock In preparation, 2026.

\bibitem{kyngLuhSong2020}
Rasmus Kyng, Kyle Luh, and Zhao Song.
\newblock Four deviations suffice for rank 1 matrices, 2020.
\newblock URL: \url{https://arxiv.org/abs/1901.06731}, \href
  {https://arxiv.org/abs/1901.06731} {\path{arXiv:1901.06731}}.

\bibitem{lehner1999}
Franz Lehner.
\newblock Computing norms of free operators with matrix coefficients.
\newblock {\em American Journal of Mathematics}, 121(3):453--486, 1999.

\bibitem{mss2015}
Adam~W. Marcus, Daniel~A. Spielman, and Nikhil Srivastava.
\newblock Interlacing families ii: Mixed characteristic polynomials and the
  kadison--singer problem.
\newblock {\em Annals of Mathematics}, 182(1):327--350, 2015.
\newblock \href {https://doi.org/10.4007/annals.2015.182.1.8}
  {\path{doi:10.4007/annals.2015.182.1.8}}.

\bibitem{pesentivladu2026}
Lucas Pesenti and Adrian Vladu.
\newblock Discrepancy minimization via regularization, 2026.
\newblock Version 2; originally appeared in SODA 2023.
\newblock URL: \url{https://arxiv.org/abs/2211.05509v2}, \href
  {https://arxiv.org/abs/2211.05509} {\path{arXiv:2211.05509}}.

\bibitem{spielmanZhang2022}
Daniel~A. Spielman and Peng Zhang.
\newblock Hardness results for {Weaver}'s discrepancy problem.
\newblock In {\em Approximation, Randomization, and Combinatorial Optimization.
  Algorithms and Techniques (APPROX/RANDOM 2022)}, volume 245 of {\em Leibniz
  International Proceedings in Informatics}, pages 40:1--40:14, 2022.
\newblock \href {https://doi.org/10.4230/LIPIcs.APPROX/RANDOM.2022.40}
  {\path{doi:10.4230/LIPIcs.APPROX/RANDOM.2022.40}}.

\bibitem{steinshakarchi2003}
Elias~M. Stein and Rami Shakarchi.
\newblock {\em Complex Analysis}, volume~2 of {\em Princeton Lectures in
  Analysis}.
\newblock Princeton University Press, 2003.

\bibitem{watrous2018}
John Watrous.
\newblock {\em The Theory of Quantum Information}.
\newblock Cambridge University Press, Cambridge, United Kingdom, 2018.
\newblock \href {https://doi.org/10.1017/9781316848142}
  {\path{doi:10.1017/9781316848142}}.

\bibitem{weaver2004}
Nik Weaver.
\newblock The {Kadison--Singer} problem in discrepancy theory.
\newblock {\em Discrete Mathematics}, 278:227--239, 2004.
\newblock \href {https://doi.org/10.1016/S0012-365X(03)00253-X}
  {\path{doi:10.1016/S0012-365X(03)00253-X}}.

\end{thebibliography}
\endgroup
\end{document}